\documentclass[11pt]{article}
\usepackage[pdftex,bookmarksopen=true,bookmarks=true,unicode,setpagesize]{hyperref}
\hypersetup{colorlinks=true,linkcolor=black,citecolor=black}

\usepackage{amssymb, amsmath, amsthm}

\allowdisplaybreaks
\usepackage{cite}

\newcommand\R{\mathbb{R}}

\newtheorem{theorem}{Theorem}[section]

\newtheorem{lemma}[theorem]{Lemma}
\newtheorem{proposition}[theorem]{Proposition}

\theoremstyle{remark}
\newtheorem{definition}[theorem]{Definition}
\theoremstyle{remark}

\theoremstyle{remark}
\newtheorem{remark}[theorem]{Remark}

\begin{document}

\begin{center}{\Large \bf
Particle-hole transformation in the continuum and determinantal point processes}
\end{center}

{\large Maryam Gharamah Ali Alshehri}\\ Department of Mathematics, Faculty of Science, University of Tabuk, Tabuk, KSA; \\
e-mail: \texttt{mgalshehri@ut.edu.sa}\vspace{2mm}

{\large Eugene Lytvynov\\ Department of Mathematics, Swansea University,  Swansea, UK;\\
e-mail: \texttt{e.lytvynov@swansea.ac.uk}\vspace{2mm}

{\small

\begin{center}
{\bf Abstract}
\end{center}

\noindent 
  Let $X$ be an underlying space with a reference measure $\sigma$.   Let $K$ be an integral operator in  $L^2(X,\sigma)$ with integral kernel $K(x,y)$. A point process $\mu$ on $X$ is called determinantal with the correlation operator  $K$ if the correlation functions of $\mu$ are given by $k^{(n)}(x_1,\dots,x_n)=\operatorname{det}[K(x_i,x_j)]_{i,j=1,\dots,n}$. It is known that each determinantal point process with a self-adjoint correlation operator $K$ is the joint spectral measure of the particle density $\rho(x)=\mathcal A^+(x)\mathcal A^-(x)$ ($x\in X$), where the operator-valued distributions $\mathcal A^+(x)$, $\mathcal A^-(x)$ come from a gauge-invariant quasi-free representation of the canonical anticommutation relations (CAR). If the space $X$ is discrete and divided  into two disjoint parts, $X_1$ and $X_2$, by exchanging particles and holes on the $X_2$ part of the space, one obtains from a determinantal point process with a self-adjoint correlation operator $K$ the determinantal point process with the $J$-self-adjoint correlation operator $\widehat K=KP_1+(1-K)P_2$. Here $P_i$ is the orthogonal projection of $L^2(X,\sigma)$ onto $L^2(X_i,\sigma)$. In the case where the space $X$ is continuous, the exchange of particles and holes makes no sense. Instead, we apply   a Bogoliubov transformation to a gauge-invariant quasi-free representation of the CAR. This transformation acts identically on the $X_1$ part of the space and exchanges the creation operators $\mathcal A^+(x)$ and the  annihilation operators $\mathcal A^-(x)$  for $x\in X_2$.  This leads to a quasi-free representation of the CAR, which is not anymore gauge-invariant. We prove that the  joint spectral measure of the corresponding particle density is the determinantal point  process with the correlation operator $\widehat K$.

\noindent  

 } \vspace{2mm}



\section{Introduction }

Let $X$ be an underlying space with a reference measure $\sigma$. (Typically $X=\R^d$ and $\sigma$ is the Lebesgue measure.) Let $\Gamma_X$ denote the space of configurations in $X$, i.e., locally finite subsets of $X$. A point process in $X$ is a probability measure on $\Gamma_X$, see e.g.\ \cite{DVJ}.  A point process $\mu$ is called {\it determinantal} if the correlation functions of $\mu$ are given by
\begin{equation}\label{bgygfyuf} k_\mu^{(n)}(x_1,\dots,x_n)=\det \big[K(x_i,x_j)\big]_{i,j=1}^n,\quad n\in\mathbb N,\end{equation} 
 see e.g.\ \cite{GY,Soshnikov,Borodin,Macchi}.
The function $K(x,y)$ is called the {\it correlation kernel of $\mu$}. To study $\mu$, one usually considers a (bounded) integral operator $K$ in the (complex) space $\mathcal H:=L^2(X,\sigma)$ with integral kernel $K(x,y)$. One calls $K$ the {\it correlation operator of $\mu$}.

Assume that the correlation kernel $K(x,y)$ is Hermitian, i.e., $K(x,y)=\overline{K(y,x)}$, equivalently 
the operator $K$ is self-adjoint. Then, the Macchi--Soshnikov theorem \cite{Macchi,Soshnikov} gives a necessary and sufficient condition of  existence of a  determinantal point process $\mu$ with the correlation kernel $K(x,y)$. 

It was shown in \cite{L,LM} (see also \cite{Borodin,TamuraIto}) that this point process $\mu$ is the joint spectral measure of the particle density of a gauge-invariant quasi-free representation of the canononical anticommutation relations (CAR). More precisely, assume that operators $\mathcal A^+(\varphi)$ and $\mathcal A^-(\varphi)$ $(\varphi\in\mathcal H)$ satisfy the CAR: 
\begin{gather} \mathcal A^-(\varphi)=\left( \mathcal A^+(\varphi)\right)^*,\label{bvys5q45}\\
\{ \mathcal A^+(\varphi), \mathcal A^+(\psi)\}=\{ \mathcal A^-(\varphi), \mathcal A^-(\psi)\}=0,\qquad \{ \mathcal A^-(\varphi), \mathcal A^+(\psi)\}=(\psi,\varphi)_{\mathcal H}.\label{CAR}
\end{gather}
Here $\{A,B\}=AB+BA$ is the anticommutator. Assume that the operators $ \mathcal A^+(\varphi)$ and $ \mathcal A^-(\varphi)$ act in $\mathcal {AF}(\mathcal H\oplus\mathcal H)$, 
the antisymmetric Fock space over $\mathcal H\oplus\mathcal H$. Let $\mathbf A$ be the CAR $*$-algebra generated by these operators, and let $\tau$ be the vacuum state on $\mathbf A$. Furthermore, assume that the operators $ \mathcal A^+(\varphi)$ and $ \mathcal A^-(\varphi)$ 
are such that the state $\tau$ is gauge-invariant quasi-free, i.e., 
\begin{equation}\label{cdtewu5wq}
\tau\big( \mathcal A^+(\varphi_m)\dotsm  \mathcal A^+(\varphi_1) \mathcal A^-(\psi_1)\dotsm  \mathcal A^-(\psi_n)\big)=\delta_{m,n}\det \big[(K\varphi_i,\psi_j)_{\mathcal H})\big]_{i,j=1}^n, \end{equation}
where $\delta_{m,n}$ is the Kronecker symbol and $K$ is the self-adjoint bounded linear operator in $\mathcal H$ satisfying
\begin{equation}\label{vcrtwsu5y4}
 (K\varphi,\psi)_{\mathcal H}=\tau\big( \mathcal A^+(\varphi) \mathcal A^-(\psi)\big),\quad \varphi,\psi\in\mathcal H. \end{equation}
See \cite{AW} or \cite[Subsection~5.2.3]{BR}. 

Define operator-valued distributions $ \mathcal A^+(x)$ and $ \mathcal A^-(x)$ on $X$ by 
\begin{equation}\label{vrts5wqqa43a}
 \mathcal A^+(\varphi)=\int_X \varphi(x) \mathcal A^+(x)\sigma(dx),\quad  \mathcal A^-(\varphi)=\int_X\overline{\varphi(x)} \mathcal A^-(x)\sigma(dx),\quad \varphi\in\mathcal H.\end{equation}
The corresponding particle density is formally defined as the operator-valued distribution $\rho(x):= \mathcal A^+(x) \mathcal A^-(x)$, and in the smeared form, 
$$\rho(\Delta)=\int_\Delta \rho(x)\sigma(dx)=\int_\Delta  \mathcal A^+(x) \mathcal A^-(x)\sigma(dx),$$
where $\Delta\subset X$ is measurable and pre-compact.
Note that, at least formally, each operator $\rho(\Delta)$ is Hermitian and for any sets $\Delta_1$ and $\Delta_2$, the operators $\rho(\Delta_1)$ and $\rho(\Delta_2)$ commute.  
The main result of \cite{LM} was that, if $K$ is an integral operator and its integral (Hermitian) kernel $K(x,y)$ is such that the corresponding determinantal point process $\mu$ exists, then  the operators $\rho(\Delta)$ are well-defined, essentially self-adjoint, commuting, and furthermore
\begin{equation}\label{yrdysz}
\tau\big(\rho(\Delta_1)\dotsm\rho(\Delta_n)\big)=\int_{\Gamma_X}\gamma(\Delta_1)\dotsm\gamma(\Delta_n)\,\mu(d\gamma),
\end{equation} 
where $\gamma(\Delta):=|\gamma\cap\Delta|$, the number of points of the configuration $\gamma$ that belong to $\Delta$. Formula \eqref{yrdysz}  states that the moments of the operators $\rho(\Delta)$ under the gauge-invariant quasi-free state state $\tau$ are equal to the moments of the determinantal point process $\mu$. It should be stressed that the set of all monomials  $\gamma(\Delta_1)\dotsm\gamma(\Delta_n)$ is total in $L^2(\Gamma_X,\mu)$.

For each measurable, pre-compact set $\Delta\subset X$, we denote by $\tilde\rho(\Delta)$ the closure of the operator $\rho(\Delta)$. (Note that the operators $\tilde\rho(\Delta)$  are self-adjoint.) 
According to \cite[Chapter~3]{BK}, formula \eqref{yrdysz} means that the determinantal point process $\mu$ is the joint spectral measure of the operators $\tilde\rho(\Delta)$. 

Let us now assume that the underlying space $X$ is divided into two disjoint parts, $X_1$ and $X_2$. For $i=1,2$, let $P_i$ denote the orthogonal projection of $\mathcal H$ onto $\mathcal H_i:=L^2(X_i,\sigma)$, and let $J:=P_1-P_2$. An {\it (indefinite) $J $-scalar product in $\mathcal H$} is defined by
$$[f, g] := (Jf, g)_{\mathcal H} = ({P}_{1}f, {P}_{1}g)_{\mathcal H}- ({P}_{2}f, {P}_{2}g)_{\mathcal H},\quad f,g\in\mathcal H,$$
see e.g.\ \cite{indefinite}. 
  A bounded linear operator operator $\mathbb  K\in\mathcal L(\mathcal H)$ is called {\it $J $-self-adjoint} if
  $[\mathbb Kf, g] = [f,  \mathbb Kg]$ for all $f, g \in H$. If an integral operator $ \mathbb K\in\mathcal L(\mathcal H)$ is $J$-self-adjoint, then its integral kernel $\mathbb K(x,y)$ is called {\it $J$-Hermitian}. The integral kernel $\mathbb  K(x,y)$ of an integral operator $\mathbb  K\in\mathcal L(\mathcal H)$ is $J$-Hermitian  if and only if  $ \mathbb K(x,y)=\overline{\mathbb  K(y,x)}$ for $x$ and $y$ that belong to the same part $X_i$ and $\mathbb  K(x,y)=-\overline{\mathbb K(y,x)}$ if $x\in X_i$, $y\in X_j$ with $i\ne j$.
  
For an arbitrary bounded linear operator $K\in\mathcal L(\mathcal H)$, we define 
$$\widehat K:=KP_1+(1-K)P_2.$$
 Note that the transformation $K\mapsto\widehat K$ is an involution.  
If an operator $K$ is self-adjoint, then $\widehat K$ is $J$-self-adjoint, and if an operator $\mathbb K\in\mathcal L(\mathcal H)$ is $J$-self-adjoint, then $\widehat {\mathbb K}$ is self-adjoint.

 A necessary and sufficient condition of existence of a  determinantal point process with a $J$-Hermitian correlation kernel $\mathbb K(x,y)$ was given in \cite{Ldeterminantal}.
We  note that $J$-Hermitian correlation kernels naturally arise in the context of asymptotic representation theory of classical groups, such as symmetric or unitary groups of growing rank, see \cite{BO1,BO2,BO3,BO4,O}.

Assume, for a moment,  that the underlying space $X$ is discrete and $\sigma$ is the counting measure. Then  any determinantal point process with a $J$-Hermitian correlation kernel can be obtained from a determinantal point process with a Hermitian correlation kernel through a particle-hole transformation. More precisely, let $\gamma$ be a configuration in $X$, i.e., $\gamma\subset X$. We define a new configuration
\begin{equation}\label{d6eu}
I\gamma:=(\gamma\cap X_1)\cup(X_2\setminus\gamma), \end{equation}
i.e., $I\gamma$ coincides with $\gamma$ in $X_1$, and with the holes of $\gamma$ (the points unoccupied by $\gamma$) in $X_2$. Note that $I:\Gamma_X\to\Gamma_X$ is an involution. It was shown in \cite{BOO} that, if $\mu$ is a determinantal point process with a  correlation operator $K$, then $I_*\mu$\,, the pushforward of $\mu$ induced by $I$, is the determinatal point process  with the  correlation operator $\widehat K$. 
 Therefore, $\mu$ is a determinantal point process with a $J$-Hermitian correlation kernel $\mathbb K(x,y)$ if and only if $\mu=I_*\nu$, where $\nu$ is the  determinantal point process with the Hermitian correlation kernel $\widehat {\mathbb 
 K}(x,y)$.  (Note that $\nu=I_*\mu$.)

If the underlying space $X$ is continuous and $\sigma$ is a non-atomic measure, then a direct generalization of the above result is impossible. Indeed, for a configuration $\gamma$ in $X$, the set $I\gamma$ defined by  \eqref{d6eu} is uncountable, hence it is not anymore a configuration in $X$. Furthermore, the Macchi--Soshnikov theorem and  \cite[Theorem~3]{Ldeterminantal} imply that, if $K(x,y)$, the correlation kernel of a determinantal point process, is either Hermitian or $J$-Hermitian, then the operator $\widehat K$ is not even an integral operator, i.e., $\widehat K(x,y)$ does not exist.

 The aim of this paper is to prove the  following main result, which involves a Bogoliubov transformation of the CAR $*$-algebra $\mathbf A$.  We refer the reader to e.g.\ \cite[Section~5.2.2]{BR} for the definition of a Bogoliubov transformation and a discussion of its properties.  
 
 \vspace{2mm}
 
\noindent {\bf Main result.} {\it Let $\mathbb K\in\mathcal L(\mathcal H)$ be a $J$-self-adjoint integral operator, and assume that the $J$-Hermitian integral kernel $\mathbb K(x,y)$ of the operator $\mathbb K$ is the correlation kernel of a determinantal point process $\mu$. 
For the self-adjoint operator $K:=\widehat {\mathbb K}$, consider the corresponding gauge-invariant quasi-free representation of the CAR, i.e., the $*$-algebra $\mathbf A$ is generated by operators $ \mathcal A^+(\varphi)$, $ \mathcal A^-(\varphi)$ in  $\mathcal {AF}(\mathcal H\oplus\mathcal H)$ that satisfy \eqref{bvys5q45}, \eqref{CAR} and, for the vacuum state $\tau$ on~$\mathbf A$, formulas  \eqref{cdtewu5wq}, \eqref{vcrtwsu5y4} hold. Define a Bogoliubov transformation of the CAR $*$-algebra $\mathbf A$ by
 \begin{equation}\label{bgyt6ew534q4}
 \big( \mathcal A^+(\varphi),\,  \mathcal A^-(\varphi)\big)\mapsto  \big(A^+(\varphi),\, A^-(\varphi)\big),\end{equation}
 where for each $\varphi\in\mathcal H$, 
 \begin{equation}\label{cxsetw5yuw}
A^+(\varphi):= \mathcal A^+(P_1\varphi)+ \mathcal A^-(P_2\mathcal C\varphi),\quad  A^-(\varphi):= \mathcal A^-(P_1\varphi)+ \mathcal A^+(P_2\mathcal C\varphi). \end{equation}
  Here $(\mathcal C\varphi)(x):=\overline{\varphi(x)}$ is the complex conjugation. (Note that the vacuum state $\tau$ on the CAR $*$-algebra $\mathbb A$ generated by the operators $A^+(\varphi)$, $A^-(\varphi)$ ($\varphi\in\mathcal H$) is still quasi-free but not anymore gauge-invariant.) Let  operator-valued distributions $A^+(x)$ and $A^-(x)$ be determined by $A^+(\varphi)$, $A^-(\varphi)$  similarly to \eqref{vrts5wqqa43a}.  Then the corresponding particle density
 \begin{equation}\label{vrs53q}
 \rho(\Delta)=\int_\Delta A^+(x)A^-(x)\,\sigma(dx)\quad\text{\rm ($\Delta\subset X$ measurable and pre-compact)}
 \end{equation}
  is a family of well-defined, essentially self-adjoint, commuting operators in  $\mathcal {AF}(\mathcal H\oplus\mathcal H)$. Furthermore, for these operators $\rho(\Delta)$ and the determinantal point process $\mu$ with the correlation kernel $\mathbb K(x,y)$, formula \eqref{yrdysz} holds. In other words, the determinantal point process $\mu$ is the joint spectral measure of the family of self-adjoint operators~$\tilde \rho(\Delta)$.\vspace{2mm}}

Note that formula \eqref{cxsetw5yuw} implies that
\begin{equation}\label{jhf76ie6}
 A^+(x)=\begin{cases}
 \mathcal A^+(x),&\text{if }x\in X_1,\\  \mathcal A^-(x),&\text{if }x\in X_2 
 \end{cases},\quad  A^-(x)=\begin{cases}
 \mathcal A^-(x),&\text{if }x\in X_1,\\  \mathcal A^+(x),&\text{if }x\in X_2. 
 \end{cases}\end{equation}
 In words, on the $X_1$ part of $X$ we use the creation and annihilation operators of the original representation of the CAR, while on the $X_2$ part of $X$ we exchange the creation and annihilation operators  of the original representation. Hence, the Bogoliubov transformation \eqref{bgyt6ew534q4}, \eqref{cxsetw5yuw}  can be thought of as a counterpart of the involution $I$ defined by~\eqref{d6eu} in the case where the space $X$ is discrete.

In fact, in our derivation of the main result, the existence of the determinantal point process~$\mu$   follows from a general theorem regarding the joint spectral measure of a family of self-adjoint commuting operators, see \cite[Theorem~1]{LM}. Hence, as a by-product of our considerations, we obtain a new proof of existence of a determinantal point process with a $J$-Hermitian correlation kernel. 

In the case of a discrete space $X$, Koshida \cite{Koshida} proved that each pfaffian point process appears (in the terminology of the present paper) as the joint spectral measure of the particle density of a quasi-free representation of the CAR. Koshida notes:  `it seems highly nontrivial if our construction can be extended to the case of continuous systems.' While the present paper does not provide a solution to this problem, it still solves it for a particular class of (non-gauge-invariant) quasi-free states. 

The paper is organized as follows. In Section~\ref{ftsa5w77}, we discuss necessary preliminaries regarding determinantal point processes,  $J$-Hermitian correlation kernels, the correlation measures and the joint spectral measure of a family of commuting self-adjoint operators, and quasi-free states on the CAR algebra.

  In Section~\ref{c6w6w2}, we employ heuristic considerations, involving formula \eqref{vrs53q}, in order to give a rigorous definition of Hermitian operators  $\rho(\Delta)$. We also prove that these operators (algebraically) commute.

 In Section~\ref{rd65w53w}, we derive rigorous formulas for the Wick (normal) product 
  \begin{equation}\label{cxe56u7i}
{:}\rho(\Delta_1)\dotsm\rho(\Delta_n){:}=\int_{\Delta_1\times\dotsm\times\Delta_n}A^+(x_n)\dotsm A^+(x_1)A^-(x_1)\dotsm A^-(x_n)\,\sigma(dx_1)\dotsm\sigma(dx_n).\end{equation}

In Section~\ref{y643rse}, we formulate the main theorem of the paper (Theorem~\ref{gftws53q}), which states that the  
 the determinantal point process $\mu$ with the $J$-Hermitian correlation kernel $ \mathbb K(x,y)$  is the joint spectral measure of the family of the commuting self-adjoint operators $\tilde\rho(\Delta)$.  We start proving this result in Section~\ref{y643rse}. 

Finally, in Section~\ref{vcyre64e3}, we prove that the operators $\rho(\Delta)$ possess correlation functions and these are given by the right-hand side of formula \eqref{bgygfyuf}  in which $K(x,y)$ is replaced by $\mathbb K(x,y)$. This concludes the proof of our main result.

  \section{Preliminaries}\label{ftsa5w77}
  
\subsection{Determinantal point processes}

Let $X$ be a locally compact Polish space,  let $\mathcal B(X)$ be the Borel $\sigma$-algebra on $X$, and let $\mathcal B_0(X)$ denote the collection of all  sets from $\mathcal B(X)$ which are pre-compact.  The {\it configuration space\/} over $X$ is defined as the set of all locally finite subsets of $X$:
$$\Gamma_X:=\{\gamma\subset X\mid \text{for all $\Delta\in\mathcal B_0(X)$ $|\gamma\cap \Delta|<\infty$}\}.$$
Here, for a set $\Lambda$, $|\Lambda|$ denotes its capacity.
Elements $\gamma\in \Gamma_X$ are called {\it configurations}. One identifies each configuration $\gamma=\{x_i\}_{i\ge1}$ with the measure $\gamma=\sum_i\delta_{x_i}$ on $X$. Here, for $x\in X$, $\delta_{x}$ denotes the Dirac measure with mass at $x$. Through this identification, one gets the embedding of $\Gamma_X$ into the space of all Radon (i.e., locally finite) measures on $X$.

The space $\Gamma_X$ is endowed with the vague topology, i.e., the weakest topology on $\Gamma_X$ with respect to which all maps $\Gamma_X\ni\gamma\mapsto \langle\gamma,f\rangle=\sum_{x\in\gamma} f(x)$, $f\in C_0(X)$, are continuous. Here $C_0(X)$ is the space of all continuous real-valued functions on $X$ with compact support. We will denote by $\mathcal B(\Gamma_X)$ the Borel $\sigma$-algebra on $\Gamma_X$.
A probability measure $\mu$ on $(\Gamma_X,\mathcal B(\Gamma_X))$ is called a {\it point process\/} on $X$. For more detail, see e.g. \cite{Kal,DVJ}.

A point process $\mu$ can be described with the help of its correlation measures. Denote $X^{(n)}:=\{(x_1,\dots,x_n)\in X^n\mid x_i\ne x_j\text{ if }i\ne j\}$. The {\it $n$-th correlation measure of} $\mu$ is the symmetric measure $\theta^{(n)}$ on $X^{(n)}$ that satisfies
\begin{equation}
\int_{\Gamma_X} \sum_{\{x_1,\dots,x_n\}\subset\gamma}f^{(n)}(x_1,\dots,x_n)\,\mu(d\gamma)\notag =
\int_{X^{(n)}}f^{(n)}(x_1,\dots,x_n)\,\theta^{(n)}(dx_1\dotsm dx_n)\label{cdtrs}
\end{equation}
for all measurable symmetric functions $f^{(n)}:X^{(n)}\to[0,\infty)$.
Let $\sigma$ be a reference Radon measure on $(X,\mathcal B(X))$.  If the correlation measure $\theta^{(n)}$ has density $k^{(n)}:X^{(n)}\to[0,\infty)$ with respect to $\frac1{n!}\,\sigma^{\otimes n}$, then $k^{(n)}$ is called the {\it $n$-th correlation function of $\mu$}.
Under a mild condition on the growth of correlation measures as $n\to\infty$, they  determine a point process uniquely \cite{Lenard}.

Recall that a point process $\mu$ is called {\it determinantal\/} if there exists a complex-valued function $K(x,y)$ on $X^2$, called the {\it correlation kernel}, such that \eqref{bgygfyuf} holds, see e.g.\ \cite{Soshnikov,Borodin}. The integral operator $K$ in the complex $L^2$-space $\mathcal H=L^2(X,\sigma)$ which has integral kernel $K(x,y)$ is called the {\it correlation operator of\/} $\mu$.

Note that, for a given integral operator $K$ in $\mathcal H$, the integral kernel $K(x,y)$ is defined up to a set of  zero measure $\sigma^{\otimes 2}$. When calculating the value of $\det \big[K(x_i,x_j)\big]_{i,j=1}^n$, one has to use the values of $K(\cdot,\cdot)$ on the diagonal $\{(x,x)\in X^2\mid x\in X\}$. However, the latter set is of zero measure $\sigma^{\otimes 2}$ if the measure  $\sigma$ is non-atomic, i.e., $\sigma(\{x\})=0$ for all $x\in X$. Hence, when speaking about the correlation operator $K$ of a determinantal point process, one has to properly choose the values of the integral kernel of $K$ on the diagonal in $X^2$.

\subsection{ $J$-Hermitian correlation kernels}\label{waaww5}

Assume that the underlying space $X$ is split into two disjoint measurable parts, $X_1$ and $X_2$, of positive measure $\sigma$. Just as in Introduction, we denote by $P_i$  the orthogonal projection of $\mathcal H$ onto $\mathcal H_i=L^2(X_i,\sigma)$, and we let  $J=P_1-P_2$.

  According to the orthogonal sum $\mathcal H=\mathcal H_1\oplus \mathcal H_2$, each operator $A\in\mathcal L(\mathcal H)$ can be represented 
  in the block form, \begin{equation}\label{fydydyf}
A=\left[\begin{matrix}
A^{11}& A^{21}\\
A^{12}&A^{22}
\end{matrix}\right],\end{equation}
where $A^{ij}:\mathcal H_j\to \mathcal H_i$, $i,j=1,2$. Here $A^{ij}:=P_iAP_j$.  Then the operator $A$ being $J$-self-adjoint means that $(A^{ii})^*=A^{ii}$ ($i=1,2$) and $(A^{21})^*=-A^{12}$. 
 

We denote by $\mathcal S_1(\mathcal H)$  the set of all trace-class operators in $\mathcal H$, and by $\mathcal S_2(\mathcal H)$ the set of all Hilbert--Schmidt operators in $\mathcal H$. For $\Delta\in\mathcal B_0(X)$, we denote by $P_\Delta$ the orthogonal projection of $\mathcal H$ onto $L^2(\Delta,\sigma)$. For $i=1,2$, we denote $\mathcal B_0(X_i):=\{\Delta\in\mathcal B_0(X)\mid\Delta\subset X_i\}$.

We say that an operator $K\in\mathcal L(\mathcal H)$ is {\it locally trace-class on $X_i$} ($i=1,2$) if, for each $\Delta_i\in\mathcal B_0(X_i)$, we have  
$K_{\Delta_i}:=P_{\Delta_i}KP_{\Delta_i}\in  \mathcal S_1(\mathcal H)$.

The following lemma can be easily checked by using  basic properties of trace-class and  Hilbert--Schmidt operators, see e.g.\ \cite{Simon}.

\begin{lemma}\label{cxtsra5aq} Let $ K\in\mathcal L(\mathcal H)$ satisfy $\mathbf 0\le  K\le\mathbf 1$. Define $K_1:=\sqrt K$, $K_2:=\sqrt{\mathbf 1- K}$. Then the following statements are equivalent.

(i) The operator $K$ is locally trace-class on $X_1$ and the operator $\mathbf 1- K$ is locally trace-class on $X_2$.

(ii) For any $\Delta_i\in\mathcal B_0(X_i)$ ($i=1,2$), we have $ K_iP_{\Delta_i}\in\mathcal S_2(\mathcal H)$, or equivalently $P_{\Delta_i}K_i\in\mathcal S_2(\mathcal H)$.
\end{lemma}

 Note that, in Lemma~\ref{cxtsra5aq}, $K_iP_i$ and $P_i  K_i$ ($i=1,2$)  are integral operators and their respective integral kernels $ K_i(x, y)$ with  $(x,y)\in (X\times X_i)\cup(X_i\times X)$ satisfy 
\begin{equation}
\int_{(X\times \Delta_i)\cup(\Delta_i\times X)}| K_i(x, y)|^2\,\sigma(dx)\,\sigma(dy) <\infty,\label{ers5w}
\end{equation}
 for any $\Delta_i\in \mathcal B_0(X_i)$. Without loss of generality, we may assume that 
 $K_i(x,y)=\overline{K_i(y,x)}$ for all $(x,y)\in (X\times X_i)\cup(X_i\times X)$, and $\int_X | K(x,y)|^2\,\sigma(dy)<\infty$ for all $x\in X_i$.
 
Now consider a $J$-self-adjoint  operator $\mathbb K\in\mathcal L(\mathcal H)$ and denote $K:=\widehat{\mathbb K}=\mathbb KP_1+\linebreak(\mathbf 1-\mathbb K)P_2$. Note that  $\mathbb K^{11}= K^{11}$, $\mathbb K^{22}= (\mathbf{ 1}-K)^{22}$, $\mathbb K^{21}=K^{21}$, and $\mathbb K^{12}=- K^{12}=-(K^{21})^*$.  Let us assume that the operator $K$ satisfies the assumptions of Lemma~\ref{cxtsra5aq}, equivalently the operator $\mathbb K$ is locally trace-class on both $X_1$ and $X_2$ and $\mathbf 0\le \widehat{\mathbb K}\le\mathbf 1$.

Let us show that $\mathbb K$ is an integral operator, and let us present an integral kernel of $\mathbb K$. For $i=1,2$, we set 
\begin{equation}\label{eraq43}
\mathbb K(x,y)=\int_X K_i(x,z) K_i(z,y)\, \sigma(dz),\quad (x,y)\in X_i^2, 
\end{equation}
which is an integral kernel of $\mathbb K^{ii}$. Note that, for all $(x,y)\in X_i^2$, we have $\mathbb K(y,x)=\overline{\mathbb K(x,y)}$. Next, for any $\Delta_i\in\mathcal B_0(X_i)$ ($i=1,2$), we have $P_{\Delta_2}\mathbb KP_{\Delta_1}= P_{\Delta_2}\mathcal \mathbb KP_{\Delta_1}\in\mathcal S_2(\mathcal H)$. Hence, 
   $\mathbb K^{21}$ is an integral operator. We choose an arbitrary integral kernel of $\mathbb K^{21}$, denoted by $\mathbb K(x,y)$ with $x\in X_2$ and $y\in X_1$.
Finally, we set $\mathbb K(x,y)=-\overline{\mathbb K(y,x)}$ for $x\in X_1$ and  $y\in X_2$. Thus, we have constructed an integral kernel of the operator $\mathbb K$.

The following theorem is shown in  \cite[Theorem~2]{Ldeterminantal}.

\begin{theorem}\label{cfyst6esa} Let a $J$-self-adjoint operator $\mathbb K\in\mathcal L(\mathcal H)$ be locally trace-class on both $X_1$ and $X_2$ and such that $\mathbf 0\le \widehat{\mathbb K}\le\mathbf 1$.  Let the integral kernel $\mathbb K(x,y)$ of the integral operator $\mathbb K$ be chosen as above.  Then there exists a unique determinantal point process with the correlation kernel $\mathbb K(x,y)$. \end{theorem}

\begin{remark} In fact, the conditions of Theorem~\ref{cfyst6esa} are necessary for the existence of a determinantal point process with a $J$-Hermitian correlation kernel, see \cite[Theorem~3]{Ldeterminantal}. 
\end{remark}

\begin{remark} While Theorem~\ref{cfyst6esa} will serve as a motivation for our studies, we will not actually use it and will derive the existence of a determinantal point process as in Theorem~\ref{cfyst6esa}  by methods different to \cite{Ldeterminantal}. Nevertheless, we will use the choice of the integral kernel $\mathbb K(x,y)$ as described above.
\end{remark}

\subsection{Joint spectral measure of a family of commuting self-adjoint operators}\label{crteswu5wu}

Let us now present a result from \cite{LM} on the joint spectral measure of  a family of commuting self-adjoint operators. Our brief presentation essentially follows \cite[Section~4]{AL}. 

Let $\mathcal F$ be a separable Hilbert space and let $\mathcal D$ be a dense subspace of $\mathcal F$.  For each $\Delta\in\mathcal B_0(X)$, let  $\rho(\Delta):\mathcal D\to\mathcal D$ be a linear Hermitian operator in $\mathcal F$. We further assume:

\begin{itemize}

\item for any $\Delta_1,\Delta_2\in\mathcal B_0(X)$ with $\Delta_1\cap\Delta_2=\varnothing$, we have $\rho(\Delta_1\cup\Delta_2)=\rho(\Delta_1)+\rho(\Delta_2)$;

\item the operators $\rho(\Delta)$ commute, i.e.,  
$[\rho(\Delta_1),\rho(\Delta_2)]=0$ for any $\Delta_1,\Delta_2\in\mathcal B_0(X)$.
\end{itemize}

Let $\mathcal A$ denote 
 the  (commutative) $*$-algebra generated by $(\rho(\Delta))_{\Delta\in\mathcal B_0(X)}$. Let $\Omega$ be a fixed vector in $\mathcal D$  with $\|\Omega\|_{\mathcal  F}=1$, and let a state $\tau:\mathcal A\to\mathbb C$ be defined by $\tau(a):=(a\Omega,\Omega)_{\mathcal  F}$ for $a\in\mathcal A$.

We define {\it Wick polynomials in $\mathcal A$} by the following recurrence formula:
\begin{align}
{:}\rho(\Delta){:}&=\rho(\Delta),\notag\\
{:}\rho(\Delta_1)\dotsm \rho(\Delta_{n+1}){:}&=\rho(\Delta_{n+1})\,
{:}\rho(\Delta_1)\dotsm \rho(\Delta_{n}){:}\notag\\
&\quad-\sum_{i=1}^n
{:}\rho(\Delta_1)\dotsm \rho(\Delta_{i-1})\rho(\Delta_i\cap\Delta_{n+1})\rho(\Delta_{i+1})\dotsm \rho(\Delta_n){:}\,,\label{dsaea78}
 \end{align}
 where $\Delta,\Delta_1,\dots,\Delta_{n+1}\in\mathcal B_0(X)$ and $n\in\mathbb N$. 
 It is easy to see that, for each permutation $\pi\in  S_n$,
$${:}\rho(\Delta_1)\cdots \rho(\Delta_n){:} = {:}\rho(\Delta_{\pi(1)})\cdots \rho(\Delta_{\pi(n)}){:}\, .$$

We assume that, for each $n\in\mathbb N$, there exists a symmetric measure $\theta^{(n)}$ on
$X^n$ that is concentrated on $X^{(n)}$ (i.e., $\theta^{(n)}(X^n\setminus X^{(n)})=0$) and 
  such that
\begin{equation}\label{6esuw6u61d}
\theta^{(n)}\big(\Delta_1\times\dots\times\Delta_n\big)=\frac1{n!}\,
\tau\big({:}\rho(\Delta_1)\dotsm \rho(\Delta_{n}){:}\big),\quad  \Delta_1,\dots,\Delta_{n}\in\mathcal B_0(X).\end{equation}
Note that, if the measure $\theta^{(n)}$ exists, then it is unique. The $\theta^{(n)}$ is called the {\it $n$-th correlation measure of the operators $\rho(\Delta)$}. If  $\theta^{(n)}$ has a density $k^{(n)}$ with respect to $\frac1{n!}\sigma^{\otimes n}$, then $k^{(n)}$ is called the {\it $n$-th correlation function of the operators $\rho(\Delta)$}.

\begin{theorem}[\!\!\cite{LM}] \label{cdyre6u3}

Let $(\rho(\Delta))_{\Delta\in\mathcal B_0(X)}$ be a family of Hermitian operators in $\mathcal  F$ as above. In particular, these operators have correlation measures $(\theta^{(n)})_{n=1}^\infty$ respective the state $\tau$.  Furthermore, we assume that the following two conditions are satisfied.

{\rm(LB1)} For each $\Delta\in\mathcal B_0(X)$, there exists a constant $C_\Delta>0$ such that
\begin{equation}\label{fte76i4}
\theta^{(n)}(\Delta^n)\le  C_\Delta^n,\quad n\in\mathbb N.
\end{equation}

{\rm(LB2)} For any sequence
    $\{\Delta_{l}\}_{l\in\mathbb{N}}\subset\mathcal{B}_{0}(X)$ such
    that $\Delta_{l}\downarrow\varnothing$ (i.e.,     $\Delta_1\supset\Delta_2\supset\Delta_3\supset\cdots$ and $\bigcap_{l=1}^\infty\Delta_l=\varnothing$), we have $C_{\Delta_{l}}\rightarrow
    0$ as $l\rightarrow\infty$.

Then the following statements hold.

{\rm (i)}  Let $\mathfrak D:=\{a\Omega\mid a\in\mathcal A\}$ and let $\mathfrak F$ denote the closure of $\mathfrak D$ in $\mathcal  F$. Each operator $(\rho(\Delta),\mathfrak D)$ is essentially self-adjoint in $\mathfrak F$, i.e., the closure  of $\rho(\Delta)$, denoted by $\widetilde \rho(\Delta)$,  is a self-adjoint operator in $\mathfrak F$. 

{\rm (ii)} For any $\Delta_1,\Delta_2\in\mathcal B_0(X)$, the projection-valued measures (resolutions of the identity) of the operators $\widetilde \rho(\Delta_1)$ and $\widetilde \rho(\Delta_2)$ commute. 

{\rm (iii)} There exist a unique point process $\mu$ on $X$ and a unique unitary operator\linebreak $U:\mathfrak F\to L^2(\Gamma_X,\mu)$ satisfying $U\Omega=1$ and 
\begin{equation}\label{te5w6u3e}
U(\rho(\Delta_1)\dotsm \rho(\Delta_{n})\Omega)=\gamma(\Delta_1)\dotsm\gamma(\Delta_n)\end{equation}
 for any $\Delta_1,\dots,\Delta_{n}\in\mathcal B_0(X)$ ($n\in\mathbb N$). In particular,
\begin{equation}\label{cts6wu4w5}
\tau\big(\rho(\Delta_1)\dotsm \rho(\Delta_{n})\big)= \int_{\Gamma_X}\gamma(\Delta_1)\dotsm\gamma(\Delta_n)\,\mu(d\gamma).\end{equation}

{\rm (iv)} The correlations measures of the point process  $\mu$ are   $(\theta^{(n)})_{n=1}^\infty$. 

  \label{vrsa456ew}
 \end{theorem}

 According to \cite[Chapter~3]{BK}, the point process  $\mu$ from Theorem~\ref{vrsa456ew} is the {\it joint spectral measure of the family of commuting self-adjoint operators $(\widetilde \rho(\Delta))_{\Delta\in\mathcal B_0(X)}$}.

\subsection{Quasi-free states on the CAR algebra}

Let $\mathcal  F$ be a separable Hilbert space, and let $a^+(\varphi)$ and $a^-(\varphi)$ ($\varphi\in\mathcal H$) be bounded linear operators in $\mathcal  F$ such that $a^+(\varphi)$ linearly depends on $\varphi$ and $a^-(\varphi)=\big(a^+(\varphi)\big)^*$. Let $a^+(\varphi)$ and $a^-(\varphi)$ satisfy the CAR, i.e., formula \eqref{CAR} holds in which the operators  $\mathcal A^+(\varphi)$, $\mathcal A^-(\varphi)$ are replaced by $a^+(\varphi)$, $a^-(\varphi)$.
 Let $\mathbb A$ be the $*$-algebra generated by these operators. We define {\it field operators} $b(\varphi):=a^+(\varphi)+a^-(\varphi)$ ($\varphi\in\mathcal H$). As easily seen, these operators also generate $\mathbb A$. 

Let  $\tau:\mathbb A\to\mathbb C$ be a state on $\mathbb A$. The state $\tau$ is completely determined by the functionals $T^{(n)}:\mathcal H^n\to\mathbb C$ $(n\in\mathbb N)$ defined by
\begin{equation}\label{hyfd6swaq3}
T^{(n)}(\varphi_1,\dots,\varphi_n):=\tau\big(b(\varphi_1)\dotsm b(\varphi_n)\big).
\end{equation}
 The state $\tau$ is called {\it quasi-free\/} if 
 \begin{gather}
 T^{(2n-1)}=0,\label{cdst6w5yw}\\
T^{(2n)}(\varphi_1,\dots,\varphi_{2n})=\sum (-1)^{\operatorname{Cross}(\nu)}\, T^{(2)}(\varphi_{i_1},\varphi_{j_1})\dotsm T^{(2)}(\varphi_{i_n},\varphi_{j_n}),\quad n\in\mathbb N ,\label{vcfts5y43}\end{gather}
where the summation  is over all partitions $\nu=\big\{\{i_1,j_1\},\dots,\{i_n,j_n\}\big\}$ of $\{1,\dots,2n\}$ with $i_k<j_k$ ($k=1,\dots,n$) and $\operatorname{Cross}(\nu)$ denotes  the number of all crossings in $\nu$, i.e., the number of all choices of $\{i_k,j_k\},\{i_l,j_l\}\in\nu$ such that $i_k<i_l<j_k<j_l$,  see e.g.\ \cite[Section~5.2.3]{BR}.

The state $\tau$ is called {\it gauge-invariant\/} if, for each $q\in\mathbb C$ with $|q|=1$, we have $T^{(n)}(q\varphi_1,\dots,q\varphi_n)=T^{(n)}(\varphi_1,\dots,\varphi_n)$ for all $\varphi_1,\dots,\varphi_n\in\mathcal H$, $n\in\mathbb N$. 
The state $\tau$ can also be uniquely characterized by the $n$-point functions $S^{(m,n)}:\mathcal H^{m+n}\to\mathbb C$ ($m+n\ge 1$) defined by
\begin{equation}\label{cftesw54}
S^{(m,n)}(\varphi_1\dots,\varphi_m,\psi_1,\dots,\psi_n):=\tau\big(a^+(\varphi_1)\dotsm a^+(\varphi_m)a^-(\psi_1)\dotsm a^-(\psi_n)\big). \end{equation}
The  state $\tau$ is gauge-invariant quasi-free if and only if
\begin{equation}\label{cxtsw64ured}
S^{(m,n)}(\varphi_m\dots,\varphi_1,\psi_1,\dots,\psi_n)=\delta_{m,n}\operatorname{det}\left[S^{(1,1)}(\varphi_i,\psi_j)\right]_{i,j=1,\dots,n}\,. \end{equation}

Let us briefly recall the Araki--Wyss \cite{AW} construction of the gauge-invariant quasi-free states. 
Let $\mathcal G$ denote a separable complex Hilbert space. Let 
$\mathcal{AF}(\mathcal G):=\bigoplus_{n=0}^\infty \mathcal G^{\wedge n}n!$ 
denote the {\it antisymmetric Fock space of $\mathcal G$}. Here $\wedge$ denotes the antysymmetric tensor product and elements of the Hilbert space $\mathcal{AF}(\mathcal G)$ are sequences $g=(g^{(n)})_{n=0}^\infty$ with $g^{(n)}\in  \mathcal G^{\wedge n}$ ($ \mathcal G^{\wedge 0}:=\mathbb C$) and $\|g\|^2_{\mathcal{AF}(\mathcal G)}=\sum_{n=0}^\infty \|g^{(n)}\|^2_{\mathcal G^{\wedge n}}\,n!<\infty$. The vector $\Omega=(1,0,0,\dots)$ is called the {\it vacuum}.

For $\varphi\in\mathcal G$, we define a {\it creation operator\/} $a^+(\varphi)\in\mathcal L\big(\mathcal{AF}(\mathcal G)\big)$ by $a^+(\varphi)g^{(n)}:=\varphi\wedge g^{(n)}$ for $g^{(n)}\in\mathcal G^{\wedge n}$. 
 For each $\varphi\in\mathcal G$, we define an {\it annihilation operator\/} $a^-(\varphi):=a^+(\varphi)^*$. Then, 
 $$a^-(\varphi)g_1\wedge\dotsm \wedge g_n=\sum_{i=1}^n (-1)^{i+1}(g_i,\varphi)_{\mathcal G}\,g_1\wedge\dotsm g_{i-1}\wedge g_{i+1}\dotsm g_n$$
 for all $g_1,\dots,g_n\in\mathcal G$.
 Note that
\begin{equation}\label{cyd7ewdf}
\|a^+(\varphi)\|_{\mathcal L(\mathcal{AF}(\mathcal G))}=\|a^-(\varphi)\|_{\mathcal L(\mathcal{AF}(\mathcal G))}=\|\varphi\|_\mathcal G\,. \end{equation}
The operators $a^+(\varphi)$, $a^-(\varphi)$ satisfy the CAR (over $\mathcal G$).

Let now $\mathcal G=\mathcal H\oplus \mathcal H$, and for $\varphi\in\mathcal H$ and $\Diamond\in\{+,-\}$, we denote $a^\Diamond _1(\varphi):=a^\Diamond (\varphi,0)$, $a^\Diamond _2(\varphi):=a^\Diamond(0,\varphi)$. 

We fix any $K\in\mathcal L(\mathcal H)$  such that  $\mathbf 0\le K\le\mathbf 1$, and define the operators $K_1$ and $K_2$ as in Lemma~\ref{cxtsra5aq}. 
We define  operators 
\begin{equation}\label{css}
\mathcal A^+(\varphi):=a^+_2( K_2\varphi)+a_1^-(\mathcal C K_1\varphi),\quad 
\mathcal A^-(\varphi):=a_2^-( K_2\varphi)+a_1^+(\mathcal C K_1\varphi),
\end{equation}
where $\mathcal C$ is the complex conjugation in $\mathcal H$. The operators $\mathcal A^+(\varphi)$, $\mathcal A^-(\varphi)$ satisfy the CAR \eqref{bvys5q45}, \eqref{CAR}. Let $\mathbf A$ denote the corresponding CAR $*$-algebra. The {\it vacuum state\/} on $\mathbf A$ is defined by $\tau(a):=(a\Omega,\Omega)_{\mathcal {AF}(\mathcal G)}$ ($a\in\mathbf A$). This state is gauge-invariant quasi-free. More exactly, setting $\mathcal F=\mathcal {AF}(\mathcal G)$ and $a^\pm(\varphi)=\mathcal A^\pm(\varphi)$, one shows that formulas \eqref{cftesw54}, \eqref{cxtsw64ured} hold, 
 with $S^{(1,1)}(\varphi,\psi)=(K\varphi,\psi)_{\mathcal H}$ ($\varphi,\psi\in\mathcal H$). In fact, each gauge-invariant quasi-free state on the CAR algebra over $\mathcal H$  can be constructed in such a way \cite{AW}.

Next, just as in Subsection~\ref{waaww5}, we assume that the space $X$ is divided into two disjoint parts, $X_1$ and $X_2$. We define operators $A^+(\varphi)$, $A^-(\varphi)$ ($\varphi\in\mathcal H$) by formula~\eqref{cxsetw5yuw}.
Hence, by \eqref{css},
\begin{align}
A^+(\varphi)&=a^+(\mathcal CK_1\mathcal CP_2 \varphi,\, K_2P_1\varphi)+a^-(\mathcal C K_1P_1\varphi,\,\mathcal CK_2 P_2 \varphi),\notag\\
A^-(\varphi)&=a^-(\mathcal CK_1\mathcal CP_2\varphi,\,K_2P_1\varphi)+a^+(\mathcal CK_1P_1\varphi,\mathcal C K_2P_2\varphi).
\label{vfr6e67ie}
\end{align}
These operators also satisfy the CAR and denote by $\mathbb A$ the corresponding CAR $*$-algebra. (Note that we have the equality of $\mathbf A$ and $\mathbb A$ as sets.) 
The vacuum state $\tau$ on $\mathbb A$ is not anymore gauge-invariant but it is still quasi-free.  More exactly, setting $\mathcal F=\mathcal {AF}(\mathcal G)$ and $a^\pm(\varphi)=A^\pm(\varphi)$, one easily shows that  formulas \eqref{hyfd6swaq3}--\eqref{vcfts5y43} hold, with 
$$T^{(2)}(\varphi,\psi)=2i\Im( K\mathbb J\varphi,\mathbb J\psi)_{\mathcal H}+(\mathbb J\psi,\mathbb J\varphi)_{\mathcal H}.$$
Here $\mathbb J\varphi:=P_1\varphi+P_2\mathcal C\varphi$.

\begin{remark} The $n$-point functions $S^{(m,n)}$ for the state $\tau$ on $\mathbb A$  can be calculated as follows. First, we note that 
\begin{align*}
S^{(1,1)}(\varphi,\psi)&=((P_1\mathbb KP_1+P_2\mathbb KP_2)\varphi,\psi)_{\mathcal H},\\
S^{(2,0)}(\varphi,\psi)&=(\psi,(\mathcal CP_2KP_1+P_1K_2\mathcal CK_2P_2)\varphi)_{\mathcal H},\\
S^{(0,2)}(\varphi,\psi)&=\overline{S^{(2,0)}(\psi,\varphi)}.
\end{align*}
Next,  if $m+n$ is odd, then $S^{(m,n)}=0$,  and if $m+n=2k$ ($k\ge 2$), then, by using Lemma~\ref{mnmnmnmn} below, we get 
$$S^{(m,n)}(\varphi_1,\dots,\varphi_{m+n})=\sum (-1)^{\operatorname{Cross}(\nu)}\, S_{i_1,\,j_1}^{(2)}(\varphi_{i_1},\varphi_{j_1})\dotsm S^{(2)}_{i_k,\,j_k}(\varphi_{i_k},\varphi_{j_k}),$$
where the summation is over all partitions $\nu=\big\{\{i_1,j_1\},\dots,\{i_k,j_k\}\big\}$ of $\{1,\dots,2k\}$ with $i_l<j_l$  and $S^{(2)}_{i_l,\,j_l}:=S^{(2,0)}$ if  $j_l\le m$, $S^{(2)}_{i_l,\,j_l}:=S^{(1,1)}$ if $i_l\le m<j_l$,  $S^{(2)}_{i_l,\,j_l}:=S^{(0,2)}$ if $i_l\ge m+1$ ($l=1,\dots,k$).
\end{remark}


\section{Rigorous construction of the particle density}\label{c6w6w2}

Let $\mathbb K\in\mathcal L(\mathcal H)$ be a $J$-self-adjoint operator satisfying the as\-sump\-tions of Theorem~\ref{cfyst6esa}, and let $K:=\widehat{\mathbb K}$.  Let $\mathcal G=\mathcal H\oplus\mathcal H$, and let the operators $A^+(\varphi),A^-(\varphi)\in\mathcal L\big(\mathcal {AF}(\mathcal G)\big)$ ($\varphi\in\mathcal H$) be defined by \eqref{vfr6e67ie}. Let the operator-valued distributions $A^+(x)$, $A^-(x)$ ($x\in X$) be be determined by $A^+(\varphi)$, $A^-(\varphi)$ similarly to \eqref{vrts5wqqa43a}. 

Recall that $\mathcal{AF}(\mathcal G)$ consists of all sequences $g=(g^{(n)})_{n=0}^\infty$ with $g^{(n)}\in  \mathcal G^{\wedge n}$ satisfying $\sum_{n=0}^\infty \|g^{(n)}\|^2_{\mathcal G^{\wedge n}}\,n!<\infty$. We denote by $\mathcal {AF}_{\mathrm{fin}}(\mathcal G)$ the dense subspace of $\mathcal {AF}(\mathcal G)$ that consists of all finite sequences $g=(g^{(n)})_{n=0}^\infty$ from $\mathcal {AF}(\mathcal G)$, i.e., for some $N\in\mathbb N$  (depending on $g$), we have $g^{(n)}=0$ for all $n\ge N$. 
We endow $\mathcal {AF}_{\mathrm{fin}}(\mathcal G)$ with the topology of the locally convex direct sum of the Hilbert spaces $\mathcal G^{\wedge n}$.

We denote by $\mathcal L(\mathcal {AF}_{\mathrm{fin}}(\mathcal G))$ the space of continuous linear operators in $\mathcal {AF}_{\mathrm{fin}}(\mathcal G)$.  Note that a linear operator $A$ acting in $\mathcal {AF}_{\mathrm{fin}}(\mathcal G)$ is continuous if and only if, for each $k$ there exists $N$ such that $A\mathcal G^{\wedge k}\subset\bigoplus_{n=0}^N\mathcal G^{\wedge n}$ and $A$ acts continuously from $\mathcal G^{\wedge k}$ into $\bigoplus_{n=0}^N\mathcal G^{\wedge n}$,  see e.g.\ \cite[Chapter II, Section~6]{SW}.

Recall the heuristic definition \eqref{vrs53q} of $\rho(\Delta)$. Our aim in this section is to rigorously define $\rho(\Delta)$ for each $\Delta\in\mathcal B_0(X)$
as an operator from $\mathcal L(\mathcal {AF}_{\mathrm{fin}}(\mathcal G))$   which is Hermitian in $\mathcal {AF}(\mathcal G)$. To this end, we start with heuristic considerations.

Think of $K_1$ and $K_2$ as self-adjoint integral operators with integral Hermitian kernels $ K_1(x,y)$ and $ K_2(x,y)$, respectively. 
Let $\Delta\in\mathcal B_0(X_1)$, and denote by $\chi_\Delta$ the indicator function of the set $\Delta$. We have
\begin{align*}
&\int_\Delta A^+(x)\,\sigma(dx)=A^+(\chi_\Delta)=a_2^+(K_2\chi_\Delta)+a_1^-(\mathcal CK_1\chi_\Delta)\\
&\quad=\int_X (K_2\chi_\Delta)(x)a_2^+(x)\,\sigma(dx)+\int_X (K_1\chi_\Delta)(x)a_1^-(x)\,\sigma(dx)\\
&\quad=\int_X\int_\Delta K_2(x,y) a_2^+(x)\,\sigma(dy)\, \sigma(dx)+\int_X\int_\Delta K_1(x,y) a_1^-(x)\,\sigma(dy)\, \sigma(dx)\\
&\quad=\int_\Delta\int_X K_2(y,x) a_2^+(y)\,\sigma(dy)\, \sigma(dx)+\int_\Delta\int_X \overline{K_1(x,y)}\, a_1^-(y)\,\sigma(dy) \,\sigma(dx)\\
&\quad=\int_\Delta\big( a_2^+(K_2(\cdot,x))+a_1^-(K_1(x,\cdot))\big)\,\sigma(dx).
\end{align*}
From here, and using similar calculations when $\Delta\in\mathcal B_0(X_i)$ ($i=1,2$), we formally conclude that
\begin{align*}
A^+(x)&=\begin{cases}
      a^+_2\big(K_2(\cdot,x)\big)+a^-_1\big(K_1(x,\cdot)\big)
& \text{if}\ x\in X_1, \\
    a^+_1\big(K_1(x,\cdot)\big)+a^-_2\big(K_2(\cdot,x)\big)
  & \text{if}\ x\in X_2,
\end{cases} \\
A^-(x)&=\begin{cases}
      a^-_2\big(K_2(\cdot,x)\big)+a^+_1\big(K_1(x,\cdot)\big)
& \text{if}\ x\in X_1, \\
    a^-_1\big(K_1(x,\cdot)\big)+a^+_2\big(K_2(\cdot,x)\big)
  & \text{if}\ x\in X_2.
\end{cases} 
\end{align*}
Denoting $\rho(x):=A^+(x)A^-(x)$, we get, for $x\in X_1$,
\begin{align}
\label{yyy}
     \rho(x)&= a^+_2\big(K_2(\cdot,x)\big)a^+_1\big(K_1(x,\cdot)\big) +a^+_2\big(K_2(\cdot,x)\big)a^-_2\big(K_2(\cdot,x)\big)   \notag\\
    &\quad+  a^-_1\big(K_1(x,\cdot)\big)a^+_1\big(K_1(x,\cdot)\big)+a^-_1\big(K_1(x,\cdot)\big)a^-_2\big(K_2(\cdot,x)\big),
\end{align}
and for $x\in X_2$,
\begin{align}
\label{ttttt}
    \rho(x)&=  a^+_1\big(K_1(x,\cdot)\big)a^+_2\big(K_2(\cdot,x)\big)+ a^+_1\big(K_1(x,\cdot)\big) a^-_1\big(K_1(x,\cdot)\big)\notag\\
    &\quad+  a^-_2\big(K_2(\cdot,x)\big) a^+_2\big(K_2(\cdot,x)\big)+a^-_2\big(K_2(\cdot,x)\big)a^-_1\big(K_1(x,\cdot)\big).
\end{align}

\begin{remark}\label{vftw6u}
Comparing formulas \eqref{yyy} and \eqref{ttttt}, we see that the right-hand side of \eqref{ttttt} can be obtained from the right-hand side of \eqref{yyy} by swapping places of $a_1^\Diamond(K_1(x,\cdot))$ and  $a_2^\Diamond(K_2(\cdot,x))$ ($\Diamond\in\{+,-\}$). 
\end{remark}

For $\Delta\in\mathcal B_0(X_1)$, let us now rigorously define $\int_X \rho(x)\sigma(dx)$.

{\it The creation and annihilation operators.} We will now define 
\begin{equation*}
\int_\Delta a^+_2\big(K_2(\cdot,x)\big)a^+_1\big(K_1(x,\cdot)\big) \sigma(dx),\quad \int_\Delta a^-_1\big(K_1(x,\cdot)\big)a^-_2\big(K_2(\cdot,x)\big) \sigma(dx).
\end{equation*}

First, we note that, for $\varphi,\psi\in \mathcal G$ and  $g^{(n)}\in \mathcal G^{\wedge n}$,
\begin{equation*}
a^+(\varphi)a^+(\psi)g^{(n)}=\varphi\wedge \psi\wedge g^{(n)}=\mathcal A_{n+2}\big(\varphi\otimes \psi\otimes g^{(n)}\big).
\end{equation*}
Here, for each $n\in\mathbb N$, we denote by $\mathcal A_n$ the antisymmetrization operator in $\mathcal G^{\otimes n}$, i.e., the orthogonal projection of $\mathcal G^{\otimes n}$ onto $\mathcal G^{\wedge n}$.

For $\varphi^{(2)}\in \mathcal G^{\otimes 2}$, we define a creation operator $a^+(\varphi^{(2)})$ in $\mathcal A\mathcal F_{\mathrm{fin}}(\mathcal G)$
by \begin{equation*}
a^+(\varphi^{(2)})g^{(n)}:= \mathcal A_{n+2}(\varphi^{(2)}\otimes g^{(n)}),\quad g^{(n)}\in\mathcal G^{\wedge n}.
\end{equation*}
Obviously $a^+(\varphi^{(2)})\in \mathcal L(\mathcal {AF}_{\mathrm{fin}}(\mathcal G))$. We also denote 
$$a^-(\varphi^{(2)}):= a^+(\varphi^{(2)})^*\upharpoonright \mathcal A\mathcal F_{\mathrm{fin}}(\mathcal G),$$
which also belongs to $\mathcal L(\mathcal {AF}_{\mathrm{fin}}(\mathcal G))$. One can easily derive explicit formulas for the action of this operator.

Thus, we formally have
\begin{equation}\label{tydses5a5aq5}
\int_\Delta a^+_2\big(K_2(\cdot,x)\big)a^+_1\big(K_1(x,\cdot)\big) \sigma(dx)=a^+\bigg(\int_\Delta \big(0,K_2(\cdot,x)\big)\otimes \big(K_1(x,\cdot),0\big)\,\sigma(dx)\bigg).\end{equation}
To give the above operator a rigorous meaning, we will first identify $\int_\Delta  K_2(\cdot,x)\otimes K_1(x,\cdot)
\,\sigma(dx) $ as an element of $\mathcal H^{\otimes 2}$.

For a linear operator $B\in\mathcal L(\mathcal H)$, denote $\overline B:=\mathcal CB\mathcal C$, i.e., $\overline B$ is the complex conjugate of $B$. If $B$ is an integral operator with integral kernel $B(x,y)$, then $\overline B$ is the integral operator with integral kernel $\overline{B(x,y)}$. In particular, if $B$ is self-adjoint, then the integral kernel of $\overline B$ is $B(y,x)$.

Now, for any $\varphi,\psi\in\mathcal H$, we formally calculate 
\begin{align}
&\bigg(\int_\Delta K_2(\cdot,x)\otimes K_1(x,\cdot)\,\sigma(dx) , \varphi\otimes \psi\bigg)_{\mathcal H^{\otimes 2}} \notag  \\
    & \quad= \int_\Delta \big(K_2(\cdot,x)\otimes K_1(x,\cdot),\varphi\otimes \psi\big)_{\mathcal H^{\otimes 2}}\, \sigma(dx)\notag \\
    &\quad =\int_\Delta\int_X K_2(y,x) \overline{\varphi(y)}\,\sigma(dy) \int_X K_1(x,z)\overline{\psi(z)}\, \sigma(dz)\, \sigma(dx) \notag\\
    &\quad =\int_X\chi_\Delta(x) \big(K_1 \overline{\psi}\big)(x) \big(\overline{K_2}\,\bar{\varphi}\big)(x)\, \sigma(dx)=
    \int_X\chi_\Delta(x)\big(K_1 \overline{\psi}\big)(x) \overline{\big( K_2 \varphi\big)(x)}\, \sigma(dx)\notag\\
    &\quad =\big(P_\Delta K_1 \overline{\psi},K_2 \varphi\big)_\mathcal H=\big(K_2P_\Delta K_1 \overline{\psi}, \varphi\big)_\mathcal H\,.\label{okijuh}
\end{align}

Since $ P_\Delta K_1\in \mathcal  S_2(H)$ (see Lemma~\ref{cxtsra5aq} (ii)), we get $K_2P_\Delta K_1\in \mathcal S_2(H)$. Therefore, $K_2P_\Delta K_1$ is an integral operator 
and we denote its integral kernel by $\big(K_2P_\Delta K_1\big)(x,y)$. Note that  $K_2P_\Delta K_1(\cdot,\cdot)\in \mathcal H^{\otimes 2}$. 
Thus, we continue \eqref{okijuh} as follows: 
$$
=\int_X\int_X \big(K_2P_\Delta K_1\big)(x,y)\overline{\psi(y)}\,\sigma(dy)\, \overline{\varphi(x)}\,\sigma(dx)   
       =\big(K_2P_\Delta K_1(\cdot,\cdot),\varphi\otimes \psi\big)_{\mathcal H^{\otimes 2}}.
$$
Hence, we rigorously define 
\begin{equation}
\label{yfyv}
\int_\Delta  K_2(\cdot,x)\otimes K_1(x,\cdot)
\, \sigma(dx) := (K_2P_\Delta K_1)(\cdot,\cdot)\in\mathcal H^{\otimes 2}.
\end{equation}

We define an isometry 
\begin{gather*}
\mathcal I_{21}:\mathcal H^{\otimes 2}\to\mathcal G^{\otimes 2}=(\mathcal H\oplus\mathcal H)\otimes(\mathcal H\oplus\mathcal H),\\
\mathcal I_{21}\varphi\otimes\psi=(0,\varphi)\otimes(\psi,0),\quad \varphi,\psi\in\mathcal H.
\end{gather*}
We denote $\big(K_2P_\Delta K_1\big)_{2,1}:=\mathcal I_{21}(K_2P_\Delta K_1)(\cdot,\cdot)$. Then, in view of \eqref{tydses5a5aq5} and \eqref{yfyv}, we define 
\begin{equation}\label{vct6a}
\int_\Delta a^+_2\big(K_2(\cdot,x)\big)a^+_1\big(K_1(x,\cdot)\big) \sigma(dx):= a^+\big(\big(K_2P_\Delta K_1\big)_{2,1}\big),\end{equation}
and so
\begin{equation}\label{vce56u3w}
\int_\Delta a^-_1\big(K_1(x,\cdot)\big)a^-_2\big(K_2(\cdot,x)\big) \sigma(dx):= a^-\big(\big(K_2P_\Delta K_1\big)_{2,1}\big).\end{equation}

{\it The neutral operator}. Our next aim is to rigorously define operators
$$\int_\Delta a^+_2\big(K_2(\cdot,x)\big)a^-_2\big(K_2(\cdot,x)\big)\sigma(dx),\quad \int_\Delta a^-_1\big(K_1(x,\cdot)\big) a^+_1\big(K_1(x,\cdot)\big)\sigma(dx).$$

For a linear operator $B\in\mathcal G$, the differential second quantization of $B$ is defined  as a  linear operator $d\Gamma(B)\in\mathcal L(\mathcal {AF}_{\mathrm{fin}}(\mathcal G))$ satisfying $d\Gamma(B)\Omega:=0$ (recall that $\Omega$ is the vacuum), and for any $g_1,\dots,g_n\in\mathcal G$,
$$d\Gamma(B)g_1\wedge\dotsm\wedge g_n=\sum_{i=1}^n g_1\wedge\dotsm\wedge g_{i-1}\wedge(Bg_i)\wedge g_{i+1}\wedge\dotsm\wedge g_n.$$
As easily seen, for any $\varphi,\psi\in\mathcal G$, we have 
$a^+(\varphi)a^-(\psi)=d\Gamma\big((\cdot,\psi)_{\mathcal G}\,\varphi\big)$, i.e., the differential second quantization of the operator $\mathcal G\ni g\mapsto(g,\psi)_{\mathcal G}\,\varphi$. Hence, we may formally write 
\begin{equation}
\label{ghfdnnt}
    \int_\Delta a^+_2\big(K_2(\cdot,x)\big) a^-_2\big(K_2(\cdot,x)\big)\,\sigma(dx)     = d\Gamma \bigg( \int_\Delta  \big(\cdot,\mathcal I_2K_2(\cdot,x)\big)_\mathcal G\, \mathcal I_2K_2(\cdot,x)\,\sigma(dx)\bigg),
\end{equation} 
where  the isometry $\mathcal I_2:\mathcal H\to\mathcal G=\mathcal H\oplus\mathcal H$ is defined by $\mathcal I_2\varphi:=(0,\varphi)$.

Let us define $\int_\Delta  \big(\cdot,K_2(\cdot,x)\big)_\mathcal H ,K_2(\cdot,x)\,d\sigma(x)$
as a bounded linear operator in $\mathcal H$. For $\varphi,\psi\in \mathcal H$, we formally calculate 
\begin{align*}
    &  \bigg( \int_\Delta  \big(\varphi,K_2(\cdot,x)\big)_\mathcal H \,K_2(\cdot,x)\,\sigma(dx),\psi\bigg)_\mathcal H  \\
    &\quad =  \int_\Delta  \big(\varphi,K_2(\cdot,x)\big)_\mathcal H \big(K_2(\cdot,x),\psi\big)_\mathcal H\, \sigma(dx)\\
    &\quad =\int_\Delta\int_X \varphi(y) \overline{K_2(y,x)}\, d\sigma(y) \int_X K_2(z,x)\overline{\psi(z)}\, \sigma(dz)\,\sigma(dx)\\
        &\quad =\int_\Delta \big(K_2\varphi\big)(x)\overline{(K_2\psi)(x)}\, \sigma(dx)=\big(K_2P_\Delta K_2 \varphi,\psi\big)_\mathcal H.
\end{align*}
Thus, we rigorously define
\begin{equation}\label{vctesawu65uw}
\int_\Delta \big(\cdot,K_2(\cdot,x)\big)_\mathcal H\,K_2(\cdot,x)\,\sigma(dx):=K_2P_\Delta K_2 .
\end{equation}
In view of \eqref{ghfdnnt} and \eqref{vctesawu65uw}, we define
\begin{equation}
 \int_\Delta a^+_2\big(K_2(\cdot,x)\big) a^-_2\big(K_2(\cdot,x)\big)\,\sigma(dx):= d\Gamma \big(\mathbf 0\oplus K_2P_\Delta K_2\big).\label{ctra6u78}
\end{equation}

Next, using the CAR, we formally have
\begin{align}
&\int_\Delta a^-_1\big(K_1(x,\cdot)\big) a^+_1\big(K_1(x,\cdot)\big)\sigma(dx)\notag\\
&\quad=- \int_\Delta a^+_1\big(K_1(x,\cdot)\big) a^-_1\big(K_1(x,\cdot)\big)\, \sigma(dx)+ \int_\Delta \big(K_1(x,\cdot),K_1(x,\cdot)\big)_\mathcal H\, \sigma(dx)\notag\\
&\quad=- \int_\Delta a^+_1\big(K_1(x,\cdot)\big) a^-_1\big(K_1(x,\cdot)\big)\, \sigma(dx)+ \int_\Delta\int_X|K_1(x,y)|^2  \sigma(dy) \, \sigma(dx).
\label{vgdstwqaa}\end{align}
Since $P_\Delta K_1$ is a Hilbert--Schmidt operator (see Lemma~\ref{cxtsra5aq}), we have
\begin{equation}\label{cxdzr}
\int_\Delta\int_X|K_1(x,y)|^2  \sigma(dy) \, \sigma(dx)=\|P_\Delta K_1\|^2_{2}\,,\end{equation}
where $\|\cdot\|_{2}$ denotes the Hilbert--Schmidt norm in $\mathcal S_2(\mathcal H)$. Noting that the operator $K_\Delta=P_\Delta K P_\Delta$ is self-adjoint, we  easily see that
\begin{equation}\label{vcyefdt32}
\|P_\Delta K_1\|^2_{2}=\operatorname{Tr}(K_\Delta)=\operatorname{Tr}(\mathbb K_\Delta).\end{equation}
Hence, by \eqref{vgdstwqaa}--\eqref{vcyefdt32} and similarly to \eqref{ctra6u78}, we rigorously define
\begin{equation}\label{tydq}
\int_\Delta a^-_1\big(K_1(x,\cdot)\big) a^+_1\big(K_1(x,\cdot)\big)\sigma(dx):=d\Gamma \big(\overline{-K_1P_\Delta K_1}\oplus \mathbf 0\big) + \operatorname{Tr}(\mathbb K_\Delta).\end{equation}

Thus, formulas \eqref{vct6a}, \eqref{vce56u3w}, \eqref{ctra6u78}, and \eqref{tydq} imply a rigorous definition of $\rho(\Delta)$ for $\Delta\in\mathcal B_0(X_1)$:
$$
\rho(\Delta):=a^+\big(\big(K_2P_{\Delta} K_1\big)_{2,1}\big)+a^-\big(\big(K_2P_{\Delta} K_1\big)_{2,1}\big)+d\Gamma \big(\overline{-K_1P_{\Delta} K_1}\oplus K_2P_{\Delta} K_2\big) + \operatorname{Tr}(\mathbb K_{\Delta}),
$$

Next, we note, if $\varphi^{(2)}\in\mathcal G^{\otimes 2}$ and $\psi^{(2)}\in\mathcal G^{\otimes 2}$ is defined by
 $\psi^{(2)}(x,y):=\varphi^{(2)}(y,x)$, then $a^+(\psi^{(2)})=-a^+(\varphi^{(2)})$. Using this observation and  Remark~\ref{vftw6u}, we similarly define $\rho(\Delta)$ for $\Delta\in\mathcal B_0(X_2)$:
$$\rho(\Delta):=- a^+\big(\big(K_2P_{\Delta} K_1\big)_{2,1}\big)- a^-\big(\big(K_2P_{\Delta} K_1\big)_{2,1}\big)  -d\Gamma \big(\overline{-K_1P_{\Delta} K_1}\oplus K_2P_{\Delta} K_2\big)+\operatorname{Tr}(\mathbb K_{\Delta}).$$
 Finally, for each $\Delta\in\mathcal B_0(X)$, we define  $\rho(\Delta):=\rho(\Delta\cap X_1)+\rho(\Delta\cap X_2)$. 

We sum up our considerations in the following definition.

\begin{definition}\label{vw5u3w} For each $\Delta\in\mathcal B_0(X)$, we define
\begin{align*}
\rho(\Delta):=&a^+\big(\big(K_2J_{\Delta} K_1\big)_{2,1}\big)+a^-\big(\big(K_2J_{\Delta} K_1\big)_{2,1}\big)\\
&+d\Gamma \big(\overline{-K_1J_{\Delta} K_1}\oplus K_2J_{\Delta} K_2\big) + \operatorname{Tr}(\mathbb K_{\Delta\cap X_1})+ \operatorname{Tr}(\mathbb K_{\Delta\cap X_2}),
\end{align*}
where  $J_\Delta:=P_{\Delta\cap X_1}-P_{\Delta\cap X_2}$. 
 Each operator $\rho(\Delta)$ belongs to $\mathcal L(\mathcal {AF}_{\mathrm{fin}}(\mathcal G))$.
\end{definition}

We note that, for each $\Delta\in\mathcal B_0(X)$,  $\rho(\Delta)$ is a densely defined, Hermitian operator in $\mathcal {AF}(\mathcal G)$.

Let $(e_i)_{i=1}^\infty$ be an orthonormal basis for $\mathcal H$ consisting of real-valued functions. The following proposition can be easily proved by using Definition~\ref{vw5u3w}.

\begin{proposition}\label{baba2323}
For $\Delta\in\mathcal{B}_0(X_1)$, we have 
\begin{align}
\label{v56}
    &\rho(\Delta)=\sum_{i,j=1}^{\infty}\Big[\big(K_2P_\Delta K_1 e_j,e_i\big)_\mathcal H\, a^+_2(e_i)a^+_1(e_j)+
    \big(K_1P_\Delta K_2 e_j,e_i\big)_\mathcal H\, a^-_1(e_i)a^-_2(e_j) \notag  \\
    &  +\big(K_2P_\Delta K_2 e_j,e_i\big)_\mathcal H\, a^+_2(e_i)a^-_2(e_j)+\big(K_1P_\Delta K_1 e_j,e_i\big)_\mathcal H\, a^-_1(e_i) a^+_1(e_j)\Big],
\end{align} 
and for $\Delta\in\mathcal{B}_0(X_2)$, we have
\begin{align}
\label{v66}
     &\rho(\Delta)=\sum_{i,j=1}^{\infty}\Big[\big(K_2P_\Delta K_1 e_i,e_j\big)_\mathcal H\, a^+_1(e_i)a^+_2(e_j)+\big(K_1P_\Delta K_2 e_i,e_j\big)_\mathcal H\, a^-_2(e_i)a^-_1(e_j) \notag\\
    &+  \big(K_1P_\Delta K_1 e_i,e_j\big)_\mathcal H\, a^+_1(e_i) a^-_1(e_j)+\big(K_2P_\Delta K_2 e_i,e_j\big)_\mathcal H\, a^-_2(e_i)a^+_2(e_j)\Big].
\end{align} 
In formulas \eqref{v56} and \eqref{v66}, the series converge strongly in $\mathcal L(\mathcal{F}_{\mathrm {fin}}(\mathcal G))$, i.e., for each $f\in \mathcal{F}_{\mathrm {fin}}(\mathcal G)$, the series applied to the vector $f$ converges in $\mathcal{F}_{\mathrm {fin}}(\mathcal G)$ (hence also in $\mathcal{F}(\mathcal G)$).
\end{proposition}

\begin{remark}
Formulas  \eqref{v56} and \eqref{v66} could serve as an alternative definition of $\rho(\Delta)$ for $\Delta$ from $\mathcal B_0(X_1)$ or $\mathcal B_0(X_2)$, respectively. However, if we initially accepted  \eqref{v56} and \eqref{v66} as the definition of $\rho(\Delta)$, it would not be {\it a priori\/} clear if such a definition   does not depend on the choice of a real orthonormal basis in $\mathcal H$. 
\end{remark}

\begin{proposition}\label{vts6e}
For any $\Delta_1,\Delta_2\in\mathcal B_0(X)$, we have $\rho(\Delta_1)\rho(\Delta_2)=\rho(\Delta_2)\rho(\Delta_1)$.
\end{proposition}

To prove Proposition~\ref{vts6e}, let us recall a result on strong convergence of bounded linear operators. Let $E_1$ and $E_2$ be  Hilbert (or even Banach) spaces. Let $(B_n)_{n=1}^\infty$ be a sequence from $\mathcal L(E_1,E_2)$ and assume that $(B_n)_{n=1}^\infty$ converges strongly to $B\in\mathcal L(E_1,E_2)$. Then, the uniform boundedness principle states that $\sup_{n\in\mathbb N}\|B_n\|<\infty$, see e.g.\ \cite[Chapter~8, Section~2]{BSU}. This immediately implies

\begin{lemma}\label{ce5w5wu5} Let $E_1,E_2,E_3$ be Hilbert spaces, let $\{B_n\}_{n=1}^\infty\in\mathcal L(E_1,E_2)$ and $\{C_n\}_{n=1}^\infty\in\mathcal L(E_2,E_3)$. Assume that the series $\sum_{n=1}^\infty B_n$ and $\sum_{n=1}^\infty$ strongly converge in  $\mathcal L(E_1,E_2)$ and $\mathcal L(E_2,E_3)$, respectively. Then the series $\sum_{m=1}^\infty\sum_{n=1}^\infty B_mC_n$ converges strongly in $\mathcal L(E_1,E_3)$.  
\end{lemma}

\begin{proof}[Proof of Proposition \ref{vts6e}] The result follows from Proposition~\ref{baba2323} and Lemma~\ref{ce5w5wu5} by a tedious calculation that uses the CAR.  
\end{proof}

\section{Wick polynomials of the particle density}\label{rd65w53w}

Recall that the Hermitian operators $\rho(\Delta)$ were heuristically defined by \eqref{vrs53q} and the corresponding Wick polynomials ${:}\rho(\Delta_1)\cdots \rho(\Delta_n){:}$ were defined by \eqref{dsaea78}. Then the heuristic formula \eqref{cxe56u7i} holds, see e.g.\ \cite[Section~C]{MS}. For the reader's convenience, we will now present a heuristic proof of  \eqref{cxe56u7i}. 

Define the operator-valued distribution ${:}\rho(x_1)\dotsm \rho(x_n){:}$ so that the following formula holds:  
$${:}\rho(\Delta_1)\dotsm\rho(\Delta_n){:}=\int_{\Delta_1\times\dotsm\times\Delta_n}{:}\rho(x_1)\dotsm \rho(x_n){:}\,\sigma(dx_1)\dotsm\sigma(dx_n)$$
for all $\Delta_1,\dots,\Delta_n\in\mathcal B_0(X)$. 
 Then, by  \eqref{dsaea78}, we have   ${:}\rho(x){:}=\rho(x)$ and
\begin{equation}\label{vyre6i54ei}
{:}\rho(x_1)\dotsm\rho(x_{n+1}){:}=\rho(x_{n+1})\,{:}\rho(x_1)\dotsm\rho(x_{n}){:}-\sum_{i=1}^n \delta(x_i,x_{n+1})\,{:}\rho(x_1)\dotsm\rho(x_{n}){:}\,.\end{equation} 
Here the generalized function $\delta(x_1,x_2)$ is defined so that
$$\int_{X^2}f^{(2)}(x_1,x_2)\delta(x_1,x_2)\,\sigma(dx_1)\sigma(dx_2):=\int_Xf^{(2)}(x,x)\,\sigma(dx).$$
It follows from the CAR that the operator-valued distributions $A^+(x)$, $A^-(x)$ satisfy the commutation relations:
\begin{equation}\label{vuyd7o}
\{A^+(x_1),A^+(x_2)\}=\{A^-(x_1),A^-(x_2)\}=0,\quad \{A^-(x_1),A^+(x_2)\}=\delta(x_1,x_2).\end{equation}
Now \eqref{vyre6i54ei} and \eqref{vuyd7o}  imply, by induction, 
\begin{equation}\label{dsesaw5uwu5}
{:}\rho(x_1)\dotsm\rho(x_n){:}=A^+(x_n)\dotsm A^+(x_1)A^-(x_1)\dotsm A^-(x_n).\end{equation}

 Note that formula \eqref{dsesaw5uwu5} does not depend on the representation of the CAR, and just states that ${:}\rho(x_1)\dotsm\rho(x_n){:}$ is the {\it Wick} ({\it normal}\/) ordering of the product $\rho(x_1)\dotsm\rho(x_n)$.

 Formula \eqref{cxe56u7i} can be recurrently written as 
\begin{gather}
{:}\rho(\Delta){:}=\rho(\Delta),\notag\\
{:}\rho(\Delta_1)\dotsm\rho(\Delta_{n}){:}=\int_{\Delta_{n}}A^+(x_{n})\,{:}\rho(\Delta_1)\dotsm\rho(\Delta_{n-1}){:}\,A^-(x_{n})\,\sigma(dx_{n}),\quad n\ge2.\label{vcte5wq}
\end{gather} 
The aim of this section is to derive a rigorous form of formula \eqref{vcte5wq}.

We start with presenting a rigorous form of the heuristic operator
\begin{equation}\label{fts5w4u}
W(\Delta,R)={:}\rho(\Delta)R{:}=\int_\Delta A^+(X)RA^-(x)\sigma(dx),\end{equation}
where  $\Delta\in\mathcal B_0(X)$ and $R\in\mathcal{L}(\mathcal{AF}_{\mathrm {fin}}(\mathcal G))$. 
The following result is inspired by Proposition~\ref{baba2323} and formula \eqref{fts5w4u}.

\begin{proposition}\label{nbutd6i}
Let $R\in\mathcal{L}(\mathcal{AF}_{\mathrm {fin}}(\mathcal G))$.
For $\Delta\in\mathcal{B}_0(X_1)$, we define 
\begin{align}
\label{v56xx}
    &W(\Delta,R):=\sum_{i,j=1}^{\infty}\Big[\big(K_2P_\Delta K_1 e_j,e_i\big)_\mathcal H\, a^+_2(e_i)Ra^+_1(e_j)+
    \big(K_1P_\Delta K_2 e_j,e_i\big)_\mathcal H\, a^-_1(e_i)Ra^-_2(e_j)\notag\\ 
  &\quad   +\big(K_2P_\Delta K_2 e_j,e_i\big)_\mathcal H\, a^+_2(e_i)Ra^-_2(e_j)
    +\big(K_1P_\Delta K_1 e_j,e_i\big)_\mathcal H\, a^-_1(e_i) Ra^+_1(e_j)\Big],
\end{align} 
and for $\Delta\in\mathcal{B}_0(X_2)$, we define
\begin{align}
     &W(\Delta,R):=\sum_{i,j=1}^{\infty}\Big[\big(K_2P_\Delta K_1 e_i,e_j\big)_\mathcal H \,a^+_1(e_i)R a^+_2(e_j)+\big(K_1P_\Delta K_2 e_i,e_j\big)_\mathcal H\, a^-_2(e_i) R a^-_1(e_j) \notag\\
     &\quad+  \big(K_1P_\Delta K_1 e_i,e_j\big)_\mathcal H\, a^+_1(e_i) R a^-_1(e_j)
     +\big(K_2P_\Delta K_2 e_i,e_j\big)_\mathcal H a^-_2(e_i)R a^+_2(e_j)\Big].\label{v66xx}
\end{align} 
For $\Delta\in\mathcal{B}_0(X)$, we define $W(\Delta,R):=W(\Delta\cap X_1,R)+W(\Delta\cap X_2,R)$. Then $W(\Delta,R)\in\mathcal{L}(\mathcal{AF}_{\mathrm {fin}}(\mathcal G))$
In formulas \eqref{v56xx} and \eqref{v66xx}, the series converge strongly in $\mathcal L(\mathcal{AF}_{\mathrm {fin}}(\mathcal G))$.

\end{proposition}

Before proving Proposition~\ref{nbutd6i}, let us present the main result of this section, which gives a rigorous form of formula \eqref{vcte5wq}.

\begin{proposition}\label{vyrsa4}
For any $\Delta_1,\dots,\Delta_{n}\in\mathcal B_0(X)$, $n\ge2$, we have
\begin{equation}\label{vtydsa}
{:}\rho(\Delta_1)\dotsm\rho(\Delta_{n}){:}=W\big(\Delta_n,{:}\rho(\Delta_
1)\dotsm\rho(\Delta_{n-1}){:}\big).
\end{equation}
\end{proposition}

Let us now prove Propositions~\ref{nbutd6i} and \ref{vyrsa4}.

\begin{proof}[Proof of Proposition~\ref{nbutd6i}]
We will present the proof only in the case $\Delta\in\mathcal B_0(X_1)$. Below, we will denote by $\mathbf 1_k$ the identity operator in $\mathcal G^{\otimes k}$.

{\it Step 1.} 
Let $R_{m,n+1}\in\mathcal{L}(\mathcal G^{\wedge (n+1)},\mathcal G^{\wedge m})$. For any $\varphi,\psi\in\mathcal G$ and $g^{(n)}\in\mathcal G^{\wedge n}$, we have
$$a^+(\varphi)R_{m,n+1}a^+(\psi)g^{(n)}=\mathcal A_{m+1}\big(\mathbf 1_1\otimes R_{m,n+1}\big) \big(\mathbf 1_1\otimes  \mathcal A_{n+1}\big) \big(\varphi\otimes \psi\otimes g^{(n)}\big).$$
We set $e_n^{(1)}:=(e_n,0)$, $e_n^{(2)}:=(0,e_n)$ ($n\in\mathbb N$), which is an orthonormal basis for $\mathcal G$. Then, for any $M,N\in\mathbb N$, 
\begin{align*}
    & \sum_{i=1}^{M}  \sum_{j=1}^{N}\big(K_2P_\Delta K_1 e_j,e_i\big)_\mathcal H\, a^+_2(e_i)R_{m,n+1}a^+_1(e_j)g^{(n)} \\
        &\quad =\mathcal A_{m+1}\big(\mathbf 1_1\otimes R_{m,n+1}\big) \big(\mathbf{ 1}_1\otimes  \mathcal A_{n+1}\big)\left[ \left(\sum_{i=1}^{M}  \sum_{j=1}^{N}\big(K_2P_\Delta K_1 e_j,e_i\big)_\mathcal H\, e_i^{(2)}\otimes e_j^{(1)} \right)\otimes g^{(n)}\right]\\
    &\quad \to \mathcal A_{m+1}\big(\mathbf 1_1\otimes R_{m,n+1}\big) \big(\mathbf{ 1}_1\otimes  \mathcal A_{n+1}\big) \left[\big(K_2P_\Delta K_1\big)_{2,1}\otimes g^{(n)}\right]
    \end{align*}
in $\mathcal G^{\wedge (m+1)}$ as $M, N\rightarrow \infty$. 
This implies that the series 
$$
\sum_{i,j=1}^{\infty}\big(K_2P_\Delta K_1 e_j,e_i\big)_\mathcal H\, a^+_2(e_i)Ra^+_1(e_j)
$$ converges strongly in $\mathcal{AF}_{\mathrm {fin}}(\mathcal G)$.

{\it Step 2.} Similarly to $\mathcal{AF}_{\mathrm {fin}}(\mathcal G)$, we define the topological vector space $\mathcal{F}_{\mathrm {fin}}(\mathcal G)$ that consists of all finite sequences $g=(g^{(n)})_{n=0}^\infty$ with $g^{(n)}\in\mathcal G^{\otimes n}$ (here $ \mathcal G^{\otimes 0}:=\mathbb C$).

For $\varphi\in \mathcal G$, we denote by $l^+(\varphi)$ the left creation operator and by $l^-(\varphi)$
the left annihilation operator  in $\mathcal{F}_{\mathrm {fin}}(\mathcal G)$:
\begin{align*}
&l^+(\varphi) g_1 \otimes g_2\otimes\dots\otimes g_n= \varphi\otimes g_1\otimes\cdots\otimes g_n, \\
&l^-(\varphi) g_1 \otimes g_2\otimes\cdots\otimes g_n=(g_1,\varphi)_\mathcal G\, g_2\otimes g_3\otimes\cdots\otimes g_n,
\end{align*}for $g_1,\dots,g_n\in \mathcal G$.
Obviously, $l^+(\varphi), l^-(\varphi)\in\mathcal{L}( \mathcal{F}_{\mathrm {fin}}(\mathcal G))$. Then, for each $g^{(n)}\in\mathcal G^{\wedge n}$,
$$a^+(\varphi)g^{(n)}=\mathcal A_{n+1}l^+(\varphi)g^{(n)},\quad a^-(\varphi)g^{(n)}=nl^-(\varphi)g^{(n)}.$$

Let $\varphi^{(2)}\in \mathcal G^{\otimes 2}$. We also define an operator $l^-(\varphi^{(2)})\in \mathcal{L}(\mathcal{F}_{\mathrm {fin}}(\mathcal G))$ by 
$$
l^-(\varphi^{(2)}) g_1 \otimes g_2\otimes\cdots\otimes g_n=( g_1\otimes g_2,\varphi^{(2)})_{\mathcal G^{\otimes 2}}\, g_3\otimes g_4\otimes\cdots\otimes g_n.
$$
In particular, if $\varphi^{(2)}=\varphi_1\otimes \varphi_2$ with $\varphi_1, \varphi_2\in \mathcal G$, then $l^-(\varphi_1\otimes \varphi_2)=l^-(\varphi_2)l^-(\varphi_1)$.

 Consider $R_{m,n-1}\in\mathcal{L}(\mathcal G^{\wedge(n-1)},\mathcal G^{\wedge m})$. Then, for $\varphi,\psi\in \mathcal G$, $g^{(n)}\in \mathcal G^{\wedge n}$, we have 
\begin{align}
    & a^-(\varphi)R_{m,n-1}a^-(\psi)  g^{(n)}
    =  m\,l^-(\varphi)R_{m,n-1} n\, l^-(\psi)  g^{(n)}\notag\\
    &\quad = mn\, l^-(\varphi)R_{m,n-1} \big(l^-(\psi)\restriction_\mathcal G \otimes  \mathbf 1_{n-1}\big) g^{(n)} =mn\, l^-(\varphi)\big(l^-(\psi)\restriction_\mathcal G \otimes R_{m,n-1}\big)g^{(n)}\notag\\
    &\quad     =mn \, l^-(\varphi) (l^-(\psi)\restriction_\mathcal G\otimes\mathbf 1_m)(\mathbf 1_{1}\otimes R_{m,n-1})g^{(n)}  =mn\,l^-(\varphi)l^-(\psi)(\mathbf 1_1\otimes R_{m,n-1}) g^{(n)}\notag\\
    &\quad =mn\,l^-(\psi\otimes \varphi)(\mathbf 1_1\otimes R_{m,n-1})g^{(n)}.\notag
     \end{align}
 Hence, for $M,N\in\mathbb N$,
\begin{align*}
    &  \sum_{i=1}^{M}  \sum_{j=1}^{N}\big(K_1P_\Delta K_2 e_j,e_i\big)_\mathcal H\, a^-_1(e_i)R_{m,n-1}a^-_2(e_j)g^{(n)}   \\
&= mn \, l^-\Big(\sum_{i=1}^{N} \sum_{j=1}^{M} \big(K_2P_\Delta K_1 e_j,e_i\big)_\mathcal H\, e_i^{(2)}\otimes e_j^{(1)}\Big)(\mathbf{1}_1\otimes R_{m,n-1})g^{(n)}\\
&\quad \to mn \, l^-\big(\big(K_2P_\Delta K_1 \big)_{2,1}\big)(\mathbf{1}_1\otimes R_{m,n-1})g^{(n)}
    \end{align*}
    in $\mathcal G^{\wedge n}$ as $M, N\to \infty$. 
This implies that the series
$$\sum_{i,j=1}^{\infty}\big(K_1P_\Delta K_2 e_j,e_i\big)_\mathcal H \,a^-_1(e_i)R_{m,n-1}a^-_2(e_j)$$
 converges strongly in $\mathcal{AF}_{\mathrm {fin}}(\mathcal G)$.

{\it Step 3.}
Let $R_{m,n-1}\in \mathcal{L}(\mathcal G^{\wedge (n-1)}, \mathcal G^{\wedge m}).$ For $\varphi,\psi\in \mathcal G$ and $g ^{(n)}\in \mathcal G^{\wedge n}$, we have
$$a^+(\varphi) R_{m,n-1} a^-(\psi)g^{(n)}  
     =n \mathcal A_{m+1}\big(\big[\big(\cdot,\psi)_\mathcal G\, \varphi\big]\otimes R_{m,n-1} \big)g^{(n)}.
$$
Hence, for any  $M, N\in\mathbb{N}$, 
   \begin{align*}
    & \sum_{i=1}^{M} \sum_{j=1}^{N}\big(K_2P_\Delta K_2 e_j,e_i\big)_\mathcal H\,  a^+_2(e_i)R_{m,n-1}a^-_2(e_j)g^{(n)}   \\
   &  \quad=n\mathcal A_{m+1} \bigg(\Big(\sum_{i=1}^{M} \sum_{j=1}^{N}\big(K_2P_\Delta K_2 e_j,e_i\big)_\mathcal H \, (\cdot,e_j^{(2)})_\mathcal G\,e_i^{(2)}\Big)\otimes  R_{m,n-1} \bigg)g^{(n)}\\
    &\quad \to n \mathcal A_{m+1} \big( (\mathbf 0\oplus K_2P_\Delta K_2)\otimes R_{m,n-1}\big) g^{(n)}
\end{align*}
in $\mathcal G^{\wedge (m+1)}$ as $M, N\to \infty$.
This implies that the series
$$\sum_{i,j=1}^{\infty}\big(K_2P_\Delta K_2 e_j,e_i\big)_\mathcal H \,a^+_2(e_i)R a^-_2(e_j)$$
 converges strongly in $\mathcal{AF}_{\mathrm {fin}}(\mathcal G)$.
 
 {\it Step 4}. Recall that, for all $x\in X_1$, we have $K_1(x,\cdot)\in \mathcal H$. Therefore, for all $x\in X_1$, the operators $a_1^+\big(K_1(x,\cdot)\big)$ and $a_1^-\big(K_1(x,\cdot)\big)$ 
are bounded in $\mathcal{AF}(\mathcal G)$ and have norm $\|K_1(x,\cdot)\|_\mathcal H$. 

Let $R_{m,n+1}\in \mathcal{L}(\mathcal G^{\wedge(n+1)},\mathcal G^{\wedge m})$. We will now show that 
\begin{equation}
\label{uuytre}
\int_{\Delta} a_1^-\big(K_1(x,\cdot)\big) R_{m,n+1} a_1^+\big(K_1(x,\cdot)\big)\,\sigma(dx)
\end{equation}
exists as a Bochner integral with values in $ \mathcal{L}(\mathcal G^{\wedge n},\mathcal G^{\wedge (m-1)})$. For the definition and properties of a Bochner integral, see e.g.\ \cite[Chapter~10, Section~3]{BSU}.

\begin{lemma}\label{trszse6}
The mappings
$$
\Delta\ni x\mapsto a_1^+\big(K_1(x,\cdot)\big)\in \mathcal{L}(\mathcal{AF}(\mathcal G)),\quad \Delta\ni x\mapsto a_1^-\big(K_1(x,\cdot)\big)\in \mathcal{L}(\mathcal{AF}(\mathcal G))
$$
are  strongly measurable. 
\end{lemma}

\begin{proof}
Consider the product space $\Delta \times X$ equipped with the corresponding $\sigma$-algebra. We will say that a function $g:\Delta \times X\to \mathbb{C}$ is 
{\it rectangle-simple\/} if it is of the form 
$$g(x,y)=\sum_{i=1}^n   c_i \chi_{A_i\times B_i}(x,y),$$ where $c_i\in\mathbb C$, $A_i\in \mathcal{B}(X_1)$, $A_i\subset\Delta$, $B_i\in \mathcal{B}(X)$, $\sigma(B_i)<\infty$ ($i=1,\dots,n$), and the sets $A_i\times B_i$ are mutually disjoint.
Since $K_1(\cdot,\cdot)\in L^2(\Delta\times X_1,\sigma^{\otimes 2})$, there exists a sequence 
$(g_n)_{n=1}^\infty$ of rectangle-simple functions such that 
$g_n(x,y)\to K_1(x,y)$ and $|g_n(x,y)|\leqslant |K_1(x,y)|$ for $\sigma^{\otimes 2}$-a.a.\ $(x,y)\in\Delta\times X$.  
In particular, $g_n(\cdot,\cdot)\to K_1(\cdot,\cdot)$ in $L^2(\Delta\times X,\sigma^{\otimes 2})$. By \eqref{cyd7ewdf}, this  implies that  $\big(a^+_1(g_n(x,\cdot))\big)_{n=1}^\infty$ and $\big(a^-_1(g_n(x,\cdot))\big)_{n=1}^\infty$ are
sequences of simple functions on $\Delta$ with values in $\mathcal{L}(\mathcal{AF}(\mathcal G))$ such that, for $\sigma$-a.a. $x\in \Delta$, 
\begin{equation*}
a^+_1(g_n(x,\cdot))\to a^+_1(K_1(x,\cdot)),\quad 
a^-_1(g_n(x,\cdot))\to a^-_1(K_1(x,\cdot)),
\end{equation*}
where convergence is in $\mathcal{L}(\mathcal{AF}(\mathcal G))$. This implies that the statement of the lemma holds. \end{proof}

Lemma~\ref{trszse6} easily implies that the mapping \begin{equation*}
\Delta\ni x\mapsto a_1^-\big(K_1(x,\cdot)\big) R_{m,n+1} a_1^+\big(K_1(x,\cdot)\big) \in \mathcal{L}(\mathcal G^{\wedge n},\mathcal G^{\wedge (m-1)})
\end{equation*}
is  strongly measurable. Furthermore, by \eqref{cyd7ewdf},
\begin{align*}
    & \int_{\Delta} \|a_1^-\big(K_1(x,\cdot)\big) R_{m,n+1} a_1^+\big(K_1(x,\cdot)\big)\|_{\mathcal{L}(\mathcal G^{\wedge n},\mathcal G^{\wedge (m-1)})} \,\sigma(dx) \\
        &\quad \le  \| R_{m,n+1} \|_{\mathcal{L}(\mathcal G^{\wedge (n+1)},\mathcal G^{\wedge m})} \int_{\Delta\times X} |K_1(x,y)|^2 \,\sigma(dx)\,\sigma(dy)< \infty.
\end{align*}
Hence, by \cite[Chapter~10, Theorem~3.1]{BSU}, the Bochner integral \eqref{uuytre} exists.

We have,
\begin{align*}
 a_1^+\big(K_1(x,\cdot)\big)& = \sum_{i=1}^\infty \int_X K_1(x,y) e_i(y)\,d\sigma(y) \, a^+_1(e_i),\\
a_1^-\big(K_1(x,\cdot)\big)& = \sum_{i=1}^\infty \int_X \overline{K_1(x,y)} e_i(y)\,d\sigma(y) \, a^-_1(e_i),
\end{align*}
 where the series converges in $\mathcal{L}(\mathcal{AF}(\mathcal G))$. Note also that, for each $N\in\mathbb{N}$, 
\begin{align*}
&\bigg\| \sum_{i=1}^N \int_X  K_1(x,y) e_i(y)\,\sigma(dy)\, a^+_1(e_i)\bigg\|_{\mathcal{L}(\mathcal{AF}(\mathcal G))}\leqslant \|K_1(x,\cdot)\|_\mathcal H,\\
&\bigg\|\sum_{i=1}^N \int_X \overline{K_1(x,y)} e_i(y)\,\sigma(dy) \, a^-_1(e_i)\bigg\|_{\mathcal{L}(\mathcal{AF}(G))}\leqslant \|K_1(x,\cdot)\|_\mathcal H.
\end{align*}
Using the dominated convergence theorem for a Bochner integral (see e.g.\ \cite[Chapter~10, Exercise~3.6]{BSU}), we get
\begin{align*}
    & \int_\Delta  a_1^-\big(K_1(x,\cdot)\big) R_{m,n+1}\, a_1^+\big(K_1(x,\cdot)\big)\,\sigma(dx)  \\
    &\quad  =\sum_{i,j=1}^\infty \int_\Delta  \sigma(dx)\int_X  \sigma(dy)\, K_1(x,y) e_j(y) \int_X  \sigma(dy^{\prime})\,  \overline{K_1(x,y^{\prime})} e_i(y^{\prime}) \,
   a_1^-(e_i) R_{m,n+1} a_1^+(e_j)\\
   &\quad  =\sum_{i,j=1}^\infty \big(K_1P_\Delta K_1 e_j, e_i\big)_\mathcal H\,
   a_1^-(e_i) R_{m,n+1} a_1^+(e_j),
\end{align*}
where the series converges in   $\mathcal{L}(\mathcal G^{\wedge n}, \mathcal G^{\wedge (m+1)})$. \end{proof}

To prove Proposition \ref{vyrsa4}, we first need the following

\begin{lemma}
\label{badr55}
Let $\Delta_1,\Delta_2\in \mathcal{B}_0(X)$ and $R\in \mathcal{L}(\mathcal{AF}_{\mathrm {fin}}(\mathcal G))$. Then 
\begin{equation}\label{yed64eaa}
\rho(\Delta_1)W(\Delta_2,R)=W(\Delta_2,\rho(\Delta_1) R)+W(\Delta_1\cap \Delta_2, R).
\end{equation}
\end{lemma}

\begin{proof} By linearity, we can assume that $\Delta_1\in\mathcal{B}_0(X_{i_{1}})$, $\Delta_2\in\mathcal{B}_0(X_{i_{2}})$, where $i_1,i_2\in\{1,2\}$. Consider, for example, the case where $\Delta_1,\Delta_2\in \mathcal B_0(X_1)$. By \eqref{v56}, \eqref{v56xx} and Lemma~\ref{ce5w5wu5}, we then write down the left-hand side of \eqref{yed64eaa} as 
\begin{align}
&\rho(\Delta_1)W(\Delta_2,R)\notag\\
&\quad=\sum_{i,j,k,l=1}^{\infty}\Big[\big(K_2P_{\Delta_1} K_1 e_j,e_i\big)_\mathcal H\, a^+_2(e_i)a^+_1(e_j)+
    \big(K_1P_{\Delta_1} K_2 e_j,e_i\big)_\mathcal H\, a^-_1(e_i)a^-_2(e_j) \notag  \\
    &\qquad  +\big(K_2P_{\Delta_1} K_2 e_j,e_i\big)_\mathcal H a^+_2(e_i)a^-_2(e_j)+\big(K_1P_{\Delta_1} K_1 e_j,e_i\big)_\mathcal H a^-_1(e_i) a^+_1(e_j)\Big]\notag\\
   &\qquad\times \Big[\big(K_2P_{\Delta_2} K_1 e_l,e_k\big)_\mathcal H\, a^+_2(e_k)Ra^+_1(e_l)+
    \big(K_1P_{\Delta_2} K_2 e_l,e_k\big)_\mathcal H\, a^-_1(e_k)Ra^-_2(e_l)\notag\\ 
  &\qquad   +\big(K_2P_{\Delta_2} K_2 e_l,e_k\big)_\mathcal H\, a^+_2(e_k)Ra^-_2(e_l)
    +\big(K_1P_{\Delta_2} K_1 e_l,e_k\big)_\mathcal H\, a^-_1(e_k) Ra^+_1(e_l)\Big],\label{xzrway4q}
\end{align}
and similarly we write down the right-hand side of  \eqref{yed64eaa}. (The appearing series converge strongly in  $\mathcal L(\mathcal{AF}_{\mathrm {fin}}(\mathcal G))$.) Through lengthy but rather straightforward calculations, one shows that both expression are equal. To give the reader a feeling  how these calculations  are carried out, we consider the following
term appearing in \eqref{xzrway4q}: 
\begin{align}\label{loks66}
          & \sum_{i,j,k,l=1}^{\infty}\big(K_1P_{\Delta_1} K_2 e_j,e_i\big)_\mathcal H\big(K_2P_{\Delta_2} K_1 e_l,e_k\big)_\mathcal H a^-_1(e_i)a^-_2(e_j)a^+_2(e_k)Ra^+_1(e_l)\notag\\
    &= \sum_{i,j,k,l=1}^{\infty}\big(K_1P_{\Delta_1} K_2 e_j,e_i\big)_\mathcal H\big(K_2P_{\Delta_2} K_1 e_l,e_k\big)_\mathcal H\, a^+_2(e_k)a^-_1(e_i)a^-_2(e_j)Ra^+_1(e_l)\notag\\
    &\quad+\sum_{i,j,k,l=1}^{\infty}\big(K_1P_{\Delta_1} K_2 e_j,e_i\big)_\mathcal H\big(K_2P_{\Delta_2} K_1 e_l,e_k\big)_\mathcal H\,   \delta_{j,k}\,a^-_1(e_i)R a^+_1(e_l),
   \end{align}
   where we used the CAR. 
The second sum on the right-hand side of  \eqref{loks66} is equal to 
\begin{align}
&\sum_{i,k,l=1}^{\infty}\big(K_1P_{\Delta_1} K_2 e_i,e_k\big)_\mathcal H\big(K_2P_{\Delta_2} K_1 e_l,e_i\big)_\mathcal  H\, a^-_1(e_k)R\, a^+_1(e_l)\notag\\
&=\sum_{k,l=1}^{\infty}\big(K_2P_{\Delta_2} K_1 e_l,K_2P_{\Delta_1} K_1 e_k\big)_\mathcal H\, a^-_1(e_k)R\, a^+_1(e_l)\notag\\
&=\sum_{k,l=1}^{\infty}\big(K_1P_{\Delta_1}(\mathbf{ 1}-K)P_{\Delta_2} K_1 e_l,e_k\big)_\mathcal H\, a^-_1(e_k)R\, a^+_1(e_l)\notag\\
&=\sum_{k,l=1}^{\infty}\big(K_1P_{\Delta_1\cap \Delta_2} K_1 e_l,e_k\big)_\mathcal H\, a^-_1(e_k)R\, a^+_1(e_l)\notag\\
&\quad-\sum_{k,l=1}^{\infty}\big(K_1P_{\Delta_1}KP_{\Delta_2} K_1 e_l,e_k\big)_\mathcal H\, a^-_1(e_k)R\, a^+_1(e_l).\label{tys4q434}
\end{align}
On the other hand, another term appearing in \eqref{xzrway4q} is:
 \begin{align}
\label{uoptrew}
         & \sum_{i,j,k,l=1}^{\infty}\big(K_1P_{\Delta_1} K_1 e_j,e_i\big)_\mathcal H \big(K_1P_{\Delta_2} K_1 e_l,e_k\big)_\mathcal H\, a^-_1(e_i) a^+_1(e_j)a^-_1(e_k) R\,a^+_1(e_l) \notag\\
    &  = \sum_{i,j,k,l=1}^{\infty}\big(K_1P_{\Delta_1} K_1 e_j,e_i\big)_\mathcal H \big(K_1P_{\Delta_2} K_1 e_l,e_k\big)_\mathcal H\, a^-_1(e_k)  a^-_1(e_i) a^+_1(e_j)R\,a^+_1(e_l) \notag\\
    &\quad+ \sum_{i,j,k,l=1}^{\infty}\big(K_1P_{\Delta_1} K_1 e_j,e_i\big)_\mathcal H \big(K_1P_{\Delta_2} K_1 e_l,e_k\big)_\mathcal H\,  \delta_{j,k}\,a^-_1(e_i) R\,a^+_1(e_l).
\end{align}
Similarly, the second sum on the right-hand side of  \eqref{uoptrew} is equal to 
\begin{align}
    &  \sum_{i,k,l=1}^{\infty}\big(K_1P_{\Delta_1} K_1 e_i,e_k\big)_\mathcal H \big(K_1P_{\Delta_2} K_1 e_l,e_i\big)_\mathcal H  a^-_1(e_k) R\,a^+_1(e_l) \notag\\
    &\quad = \sum_{k,l=1}^{\infty}\big(K_1P_{\Delta_2} K_1 e_l,K_1P_{\Delta_1} K_1 e_k\big)_\mathcal Ha^-_1(e_k) R\,a^+_1(e_l) \notag\\
    &\quad =\sum_{k,l=1}^{\infty}\big(K_1P_{\Delta_1} K P_{\Delta_2} K_1 e_l, e_k\big)_\mathcal H a^-_1(e_k) R\,a^+_1(e_l).\label{fyrsa4ftr}
\end{align}
Thus, we see that, in formula \eqref{loks66}, the `wrong' term given by  \eqref{fyrsa4ftr} cancels out, the first sum on the right-hand side of \eqref{loks66} and the first sum on the right-hand side of~\eqref{uoptrew} come from $W(\Delta_2,\rho(\Delta_1) R)$, and the first sum on the right-hand side of \eqref{tys4q434} comes from $W(\Delta_1\cap \Delta_2, R)$.

We leave the rest of calculations to the interested reader.
\end{proof}

\begin{proof}[Proof of Proposition \ref{vyrsa4}] For $\Delta_1,\Delta_2\in\mathcal B_0(X)$, we have by Lemma~\ref{badr55},
$$\rho(\Delta_1)\rho(\Delta_2)=\rho(\Delta_1)W(\Delta_2,\mathbf 1)=W(\Delta_2,\rho(\Delta_1))+\rho(\Delta_1\cap\Delta_2).$$
Hence, by  \eqref{dsaea78},
$${:}\rho(\Delta_1)\rho(\Delta_2){:}=W(\Delta_2,\rho(\Delta_1)),$$
i.e., formula \eqref{vtydsa} holds for $n=2$. Assume formula \eqref{vtydsa} holds for $n$ and let us prove it for $n+1$. By Lemma~\ref{badr55} and  \eqref{dsaea78},
\begin{align*}
&\rho(\Delta_1){:}\rho(\Delta_2)\dotsm\rho(\Delta_{n+1}){:}=\rho(\Delta_1) W\big(\Delta_{n+1},{:}\rho(\Delta_2)\dotsm\rho(\Delta_{n}){:}\big)\\
&\quad = W\big(\Delta_{n+1}, \rho(\Delta_1){:}\rho(\Delta_2)\dotsm\rho(\Delta_{n}){:}\big)+W\big(\Delta_1\cap \Delta_{n+1},{:}\rho(\Delta_2)\dotsm\rho(\Delta_{n}){:}\big)\\
&\quad= W\big(\Delta_{n+1}, {:}\rho(\Delta_1)\rho(\Delta_2)\dotsm\rho(\Delta_{n}){:}\big)\\
&\quad+\sum_{i=2}^n  W\big(\Delta_{n+1},{:}\rho(\Delta_2)\dotsm\rho(\Delta_1\cap\Delta_i)\dotsm\rho(\Delta_n){:}\big)+{:}\rho(\Delta_1\cap\Delta_{n+1})\rho(\Delta_2)\dotsm\rho(\Delta_n){:}\\
&\quad= W\big(\Delta_{n+1}, {:}\rho(\Delta_1)\rho(\Delta_2)\dotsm\rho(\Delta_{n}){:}\big)+\sum_{i=2}^{n+1} {:}\rho(\Delta_2)\dotsm\rho(\Delta_1\cap\Delta_i)\dotsm\rho(\Delta_{n+1}){:},
\end{align*}
which, by \eqref{dsaea78}, implies that formula \eqref{vtydsa} holds for $n+1$.
\end{proof}

\section{The joint spectral measure of the particle density} \label{y643rse}

Our aim now is to apply the results of Subsection~\ref{crteswu5wu} by setting $\mathcal F=\mathcal{AF}(\mathcal G)$, $\mathcal D=\mathcal{AF}_{\mathrm{fin}}(\mathcal G)$, 
$\mathcal A$ to be  the  commutative $*$-algebra generated by the particle density\linebreak $(\rho(\Delta))_{\Delta\in\mathcal B_0(X)}$, and $\tau$ to be the vacuum state on $\mathcal A$, i.e., $\tau(a)=(a\Omega,\Omega)_{\mathcal{AF}(\mathcal G)}$ ($a\in \mathcal A$), where $\Omega$ is the vacuum in $\mathcal{AF}(\mathcal G)$. 
 
 Recall the construction of the integral kernel $\mathbb K(x,y)$ of the operator $\mathbb K$ in Subsection~\ref{waaww5}. 
 
 \begin{theorem}\label{tw5yw} The operators $(\rho(\Delta))_{\Delta\in\mathcal B_0(X)}$ have correlation measures $(\theta^{(n)})_{n=1}^\infty$ respective the vacuum state state $\tau$. Furthermore, the corresponding correlation functions are given by 
 \begin{equation}\label{xdtea5w3u}
 k^{(n)}(x_1,\dots,x_n)=\det\big[\mathbb K(x_i,x_j)\big]_{i,j=1,\dots,n}.\end{equation}
 \end{theorem}

We will prove Theorem~\ref{tw5yw} below, in Section~\ref{vcyre64e3}, but let us now formulate and prove the main theorem of the paper.

\begin{theorem} \label{gftws53q} The operators $(\rho(\Delta))_{\Delta\in\mathcal B_0(X)}$, together with their correlation measures $(\theta^{(n)})_{n=1}^\infty$, satisfy the assumptions of Theorem~\ref{cdyre6u3}. Thus, the  following statements hold:

{\rm (i)}  Let $\mathfrak D:=\{a\Omega\mid a\in\mathcal A\}$ and let $\mathfrak F$ denote the closure of $\mathfrak D$ in $\mathcal {AF}(\mathcal G)$. Each operator $(\rho(\Delta),\mathfrak D)$ is essentially self-adjoint in $\mathfrak F$, i.e., the closure  of $\rho(\Delta)$, denoted by $\widetilde \rho(\Delta)$,  is a self-adjoint operator in $\mathfrak F$. 

{\rm (ii)} For any $\Delta_1,\Delta_2\in\mathcal B_0(X)$, the projection-valued measures (resolutions of the identity) of the operators $\widetilde \rho(\Delta_1)$ and $\widetilde \rho(\Delta_2)$ commute. 

{\rm (iii)} There exist a unique point process $\mu$ in $X$ and a unique unitary operator\linebreak $U:\mathfrak F\to L^2(\Gamma_X,\mu)$ satisfying $U\Omega=1$ and \eqref{te5w6u3e}. In particular, \eqref{cts6wu4w5} holds.

{\rm (iv)} The correlations functions  of the point process  $\mu$ are  given by \eqref{xdtea5w3u}.
 
 \end{theorem}
 
 \begin{proof} 
 We first prove the following
 
 \begin{lemma}\label{uftew}
  For any $\Delta\in\mathcal B_0(X)$,
 \begin{equation}\label{yrdsw5q}
 \int_\Delta\mathbb K(x,x)\sigma(dx)=\operatorname{Tr}(\mathbb K_{\Delta\cap X_1})+\operatorname{Tr}(\mathbb K_{\Delta\cap X_2}),
 \end{equation}
 and for any $\Delta_1,\dots,\Delta_n\in\mathcal B_0(X)$ ($n\ge2$), we have 
\begin{equation}\label{vcyrd6de}
 \int_{\Delta_1\times\dots\times\Delta_{n}}\det\big[\mathbb K(x_i,x_j)]_{i,j=1,\dots,n} \,\sigma(dx_1)\dotsm \sigma(dx_{n})=\sum_{\xi\in S_{n}}(-1)^{\operatorname{sgn}(\xi)}\prod_{\psi\in \operatorname{Cycles}(\xi)} \mathbb{T}_{\psi},\end{equation}
 where, for a cycle $\psi=(l_1 l_2 \dotsm l_k)$ in a permutation $\xi\in S_{n}$, we have
 \begin{align}
  \mathbb T_\psi &= \int_{\Delta_{l_1}}\mathbb K(x,x)\sigma(dx)=\operatorname{Tr}(\mathbb K_{\Delta_{l_1}\cap X_1})+\operatorname{Tr}(\mathbb K_{\Delta_{l_1}\cap X_2}),\quad\text{if }k=1,\label{vtyde6e6x}\\
 \mathbb T_\psi& =\operatorname{Tr}\big(P_{\Delta_{l_1}}\mathbb K P_{\Delta_{l_2}} \mathbb K P_{\Delta_{l_3}}\mathbb K \dotsm \mathbb K P_{\Delta_{l_k}}\mathbb K P_{\Delta_{l_1}}\big),\quad \text{if }k\ge2.\label{dsts5a}
 \end{align}
 \end{lemma}
 
 \begin{proof} Recall that that, if $S, T\in\mathcal S_2(\mathcal H)$ with integral kernels $S(x,y)$ and $T(x,y)$, respectively, then 
$$\operatorname{Tr}(ST)=\int_X (ST)(x,x)\sigma(dx)=\int_{X^2}S(x,y)T(y,x)\sigma(dx)\,\sigma(dy).$$
Hence, formula \eqref{yrdsw5q} holds by statement (ii) of Lemma~\ref{cxtsra5aq} and the construction of the integral kernel $\mathbb K(x,y)$. 

Next, we note that, for any $\Delta_1,\Delta_2\in\mathcal B_0(X)$, we have $P_{\Delta_1}\mathbb K P_{\Delta_2}\in\mathcal S_2(\mathcal H)$, hence the operator appearing in \eqref{dsts5a} is indeed of trace class. 
Formula \eqref{vcyrd6de} obviously holds with $\mathbb T_\psi$ given, for $k=1$ by \eqref{vtyde6e6x}, and for $k\ge2$,  
\begin{align*}
\mathbb T_\psi& = \int_{\Delta_{l_1}\times \Delta_{l_2}\times\dots\times\Delta_{l_{k}}} \mathbb K(x_{l_1},x_{l_2}) \mathbb K(x_{l_2},x_{l_3})\notag\\
&\quad \times \dotsm \times \mathbb K(x_{l_{k-1}},x_{l_k})\mathbb K(x_{l_k},x_{l_1})\sigma(dx_{l_1})\dotsm \sigma(dx_{l_k})\notag\\
& =\int_{X^k} \big(P_{\Delta_{l_1}}\mathbb K P_{\Delta_{l_2}}\big)(x_1,x_2) \big(P_{\Delta_{l_2}}\mathbb K P_{\Delta_{l_3}}\big)(x_2,x_3)\notag\\
& \quad \times\dotsm\times  \big(P_{\Delta_{l_{k-1}}}\mathbb K P_{\Delta_{l_k}}\big)(x_{k-1},x_k)\big(P_{\Delta_{l_k}}\mathbb K P_{\Delta_{l_1}}\big)(x_k,x_1)\sigma(dx_{1})\dotsm \sigma(dx_{k})\notag\\
&=\int_X \big(P_{\Delta_{l_1}}\mathbb K P_{\Delta_{l_2}} \mathbb K P_{\Delta_{l_3}}\mathbb K \dotsm \mathbb K P_{\Delta_{l_k}}\mathbb K P_{\Delta_{l_1}}\big)(x,x)\sigma(dx)\notag\\
 & =\operatorname{Tr}\big(P_{\Delta_{l_1}}\mathbb K P_{\Delta_{l_2}} \mathbb K P_{\Delta_{l_3}}\mathbb K \dotsm \mathbb K P_{\Delta_{l_k}}\mathbb K P_{\Delta_{l_1}}\big).\qedhere
\end{align*} 
 \end{proof}
 
 Let us prove that condition (LB1) of Theorem~\ref{cdyre6u3} is satisfied. For $l\ge2$, we have
\begin{equation}\label{xtsw5u}
\big|\operatorname{Tr}(\mathbb K^l_\Delta)\big|\le\|\mathbb K^l_\Delta\|_1
\le\|\mathbb K_\Delta\|_2\,\|\mathbb K^{l-1}_\Delta\|_2\le\|\mathbb K_\Delta\|_2^2 \,\|\mathbb K^{l-2}_\Delta\|\le\max\big\{\|\mathbb K_\Delta\|_2,\,\|\mathbb K_\Delta\|\big\}^l,
\end{equation} 
where $\|\cdot\|_1$   denotes the norm in $\mathcal S_1(\mathcal H)$.
  Theorem~\ref{tw5yw}, Lemma~\ref{uftew} and formula \eqref{xtsw5u} imply that condition (LB1) is satisfied with
 $$C_\Delta= \max\big\{\|\mathbb K_{\Delta\cap X_1}\|_1+\|\mathbb K_{\Delta\cap X_2}\|_1,\|\mathbb K_\Delta\|_2,\,\|\mathbb K_\Delta\|\big\}.$$ 
  
In the case where $\Delta$ is a subset of either $X_1$ or $X_2$, we can find a finer estimate of   $\theta^{(n)}(\Delta^n)$. Indeed, assume, for example that $\Delta\subset X_1$. Then, for any $x,y\in \Delta$, we have, by \eqref{eraq43},
$$\mathbb K(x,y)=(K_1(x,\cdot),K_1(y,\cdot))_{\mathcal H}, $$
and so, for any $x_1,\dots,x_n\in\Delta$, 
$$\det\big[\mathbb K(x_i,x_j)]_{i,j=1,\dots,n}=\big(K_1(x_1,\cdot)\wedge\dots\wedge K_1(x_n,\cdot), 
K_1(x_1,\cdot)\wedge\dots\wedge K_1(x_n,\cdot)\big)_{\mathcal H^{\wedge n}}\, n!\,,$$
which implies
$$\left|\det\big[\mathbb K(x_i,x_j)]_{i,j=1,\dots,n}\right|\le \|K_1(x_1,\cdot)\|^2_{\mathcal H}\dotsm \|K_1(x_n,\cdot)\|^2_{\mathcal H}\,n!\,.$$
Hence, by Theorem~\ref{tw5yw}, condition (LB1) is satisfied with
\begin{equation}\label{vyrsu5w}
C_\Delta=\int_{\Delta\times X}|K_1(x,y)|^2\sigma(dx)\sigma(dy).\end{equation}
Similarly, for $\Delta\in\mathcal B_0(X_2)$, (LB1) is satisfied with  
\begin{equation}\label{vegegeu5w}
C_\Delta=\int_{\Delta\times X}|K_2(x,y)|^2\sigma(dx)\sigma(dy).\end{equation}
  
 It follows from the proof of Theorem~\ref{cdyre6u3} in \cite{LM} that, when checking condition  (LB2), it is sufficient to assume that all sets in the sequence $\{\Delta_{l}\}_{l\in\mathbb{N}}$ are subsets of either $X_1$ or $X_2$. But then (LB2) is an immediate consequence of formulas 
 \eqref{vyrsu5w} and \eqref{vegegeu5w}.
 \end{proof}

\section{Proof of Theorem~\ref{tw5yw}}\label{vcyre64e3}

We will now prove Theorem~\ref{tw5yw}. Our strategy here is to prove the existence of correlation measures and that of correlation functions at the same time. 

We first state, for any $x_1,\dots,x_n\in X$,
\begin{equation}\label{ctsa53y}
\det \big[\mathbb K(x_i,x_j)\big]_{i,j=1,\dots,n}\ge0.
\end{equation}
Indeed, if all points $x_1,\dots,x_n$ belong to the same part, $X_1$ or $X_2$, then the matrix  $\big[\mathbb K(x_i,x_j)\big]_{i,j=1,\dots,n}$ is Hermitian, hence its determinant is $\ge0$. Otherwise, without loss of generality, we may assume that for some $m$ with $1<m<n$, we have $x_1,\dots,x_m\in X_1$ and $x_{m+1},\dots,x_n\in X_2$. But then formula \eqref{ctsa53y} follows from \cite[Proposition~1.4]{O}.

Next, it is easy to see that, if among points $x_1,\dots,x_n$, at least two points coincide, then $\det \big[\mathbb K(x_i,x_j)\big]_{i,j=1,\dots,n}=0$.
Therefore, the measure
$$\det \big[\mathbb K(x_i,x_j)\big]_{i,j=1,\dots,n}\,\frac1{n!}\,\sigma(dx_1)\dotsm\sigma(dx_n)$$
is concentrated on $X^{(n)}$.

Hence, to prove Theorem~\ref{tw5yw}, it suffices to show that, for any $\Delta_1,\dots,\Delta_m\in \mathcal{B}_0(X_1)$, $\Delta_{m+1},\dots,\Delta_{m+n}\in\mathcal{B}_0(X_2)$, $m,n\in\mathbb N_0$, $m+n\ge1$, we have
\begin{align}
&\tau\big({:}\rho(\Delta_1)\cdots\rho(\Delta_{m+n}){:}\big)\notag\\
&\quad=\int_{\Delta_1\times\cdots\times\Delta_{m+n}}\det\big[\mathbb K(x_i,x_j)\big]_{i,j=1,\cdots,m+n}\, \sigma(dx_1)\dotsm \sigma(dx_{m+n}).\label{vcyrtse6u785}
\end{align}

We divide the proof of this formula into several steps.

{\it Step 1}. To shorten our notations, we denote, for $i,j\in\mathbb N$ and $\Delta\in\mathcal B_0(X_1)$,
\begin{align*}
    &c_{ij}^{++}(\Delta)=\big(K_2P_\Delta K_1 e_j,e_i\big)_\mathcal H,\quad c_{ij}^{--}(\Delta)=\big(K_1P_\Delta K_2 e_j,e_i\big)_\mathcal H,  \\
    &c_{ij}^{+-}(\Delta)  =\big(K_2P_\Delta K_2 e_j,e_i\big)_\mathcal H,\quad  c_{ij}^{-+}(\Delta)=\big(K_1P_\Delta K_1 e_j,e_i\big)_\mathcal H ,
\end{align*}
for $\Delta\in\mathcal{B}_{0}(X_2)$,
\begin{align*}
&c_{ij}^{++}(\Delta)=\big(K_2P_\Delta K_1 e_i,e_j\big)_\mathcal H,\quad 
c_{ij}^{--}(\Delta)=\big(K_1P_\Delta K_2 e_i,e_j\big)_\mathcal H ,\\
&c_{ij}^{+-}(\Delta)= \big(K_1P_\Delta K_1 e_i,e_j\big)_\mathcal H ,\quad  c_{ij}^{-+}(\Delta)=\big(K_2P_\Delta K_2 e_i,e_j\big)_\mathcal H,
\end{align*}
and     
$$
     A^+_i:=a^+_2(e_i),\quad A^-_i:=a^-_1(e_i),\quad B_i^+:=a_1^+(e_i),\quad 
    B_i^-:=a_2^-(e_i).
   $$
   Then, by Proposition \ref{baba2323}, for $\Delta\in \mathcal{B}_0(X_1)$,
\begin{align*}
    &\rho(\Delta)=\sum_{i,j=1}^\infty\quad \sum_{{\Diamond_1,\Diamond_2}\in\{+,-\}} c_{ij}^{ \Diamond_1\Diamond_2}(\Delta) A^{\Diamond_1}_i B^{\Diamond_2}_j  ,
\end{align*}
and for $\Delta\in \mathcal{B}_0(X_2)$,
\begin{align*}
    &\rho(\Delta)=\sum_{i,j=1}^\infty\quad \sum_{{\Diamond_1,\Diamond_2}\in\{+,-\}} c_{ij}^{ \Diamond_1\Diamond_2}(\Delta) B^{\Diamond_1}_i A^{\Diamond_2}_j .\end{align*}
    
 We define an ordered set
\begin{equation}\label{ytde7er}
\mathfrak{E}:=\{1,\dots,m,m+1,\dots,m+n,(m+n)',\dots,(m+1)',m',(m-1)',\dots,1'\}\end{equation}
(the elements of $\mathfrak E$ being listed in \eqref{ytde7er} in the increasing order). By Proposition \ref{vyrsa4},
\begin{align}
    & \tau\big({:}\rho(\Delta_1)\cdots \rho(\Delta_{m+n}){:}\big) \notag \\
    &=\sum_{i_1,\dots,i_{m+n},i_{(m+n)'},\dots,i_{1'}\in\mathbb N}\ \sum_{\Diamond_1,\dots,\Diamond_{m+n},\Diamond_{(m+n)'},\dots,\Diamond_{1'}\in\{+,-\}}
     c_{i_1i_{1'}}^{\Diamond_1\Diamond_{1'}}(\Delta_1)\dotsm c_{i_{m+n}i_{(m+n)'}}^{\Diamond_{m+n}\Diamond_{(m+n)'}}(\Delta_{m+n})\notag\\
    &\times \tau\big(A_{i_1}^{\Diamond_1}\dotsm A_{i_m}^{\Diamond_m}B_{i_{m+1}}^{\Diamond_{m+1}}\dotsm B_{i_{m+n}}^{\Diamond_{m+n} }A_{i_{(m+n)'}}^{\Diamond_{(m+n)'}}\dotsm A_{i_{(m+1)'}}^{\Diamond(m+1)'}B_{i_{m'}}^{\Diamond_{m'}}\dotsm B_{i_{1'}}^{\Diamond_{1'}}\big)\notag\\
    &=\sum_{i_1,\dots,i_{m+n},i_{(m+n)'},\dots,i_{1'}\in\mathbb N}\quad c_{i_1i_1'}^{- +}(\Delta_1)\dotsm c_{i_{m}i_m'}^{- +}(\Delta_{m})\notag\\
&\times\sum_{\Diamond_{m+1},\dots,\Diamond_{m+n},\Diamond_{(m+n)'},\dots,\Diamond_{(m+1)'}\in\{+,-\}} c_{i_{m+1}i_{(m+1)'}}^{\Diamond_{m+1}\Diamond_{(m+1)'}}(\Delta_{m+1})\dotsm c_{i_{m+n}i_{(m+n)'}}^{\Diamond_{m+n}\Diamond_{(m+n)'}}(\Delta_{m+n})\notag \\ &\times
 \tau\big(A_{i_1}^{-}\dotsm A_{i_m}^{-}B_{i_{m+1}}^{\Diamond_{m+1}}\dotsm B_{i_{m+n}}^{\Diamond_{m+n} }A_{i_{(m+n)'}}^{\Diamond_{(m+n)'}}\dotsm A_{i_{(m+1)'}}^{\Diamond(m+1)'}B_{i_{m'}}^{+}\dotsm B_{i_{1'}}^{+}\big).\label{hjyu09}
\end{align}\vspace{2mm}

{\it Step 2}. We remind the reader that $\tau$ is a quasi-free state. As easily seen, the following lemma holds. 
  
\begin{lemma} \label{mnmnmnmn}
Let $n\in\mathbb{N}$  and let $g_1,\dots,g_{2n}\in \mathcal G$. Let $\Diamond_1,\dots,\Diamond_{2n}\in\{+,-\}$ and  assume that the number of plusses among $\Diamond_1,\dots,\Diamond_{2n}$  is the same as the number of minuses. 
Then, 
\begin{equation*}
\big(a^{\Diamond_1}(g_1)\cdots a^{\Diamond_n}(g_{2n})\Omega,\Omega\big)_{\mathcal{AF}(\mathcal G)} 
    =\sum (-1)^{\operatorname{Cross}(\nu)}
    \bigg(\prod_{\substack{\{i,j\}\in\nu\\ i< j}}(g_j,g_i)_\mathcal G\bigg),
\end{equation*}
where the summation is over  all partitions $\nu=\big\{\{i_1,j_1\},\dots,\{i_{n},j_{n}\}\big\}$ of $\{1,\dots,2n\}$ with $i_k<j_k$ and such that $\Diamond_{i_k}=-$, $\Diamond_{j_k}=+$ ($k=1,\dots,n$). 
\end{lemma}

We define  four (ordered) subsets of $\mathfrak E$ as follows: 
\begin{gather}
\mathfrak A:=\{1,\dots,m\},\quad \mathfrak B:=\{m+1,\dots,m+n\},\notag\\
\mathfrak C:=\{(m+n)',\dots,(m+1)'\},\ \mathfrak D:=\{m',\dots,1'\}.\label{yut8}
\end{gather}
Denote by $\mathfrak R$ the collection of all partitions $\nu$ of $\mathfrak E$ into $n$ two-point sets such that, if $\{i,j\}\in\nu$ with $i<j$, 
 then one of the following four statements holds: (i) $i\in \mathfrak A$ and $j\in \mathfrak B$; (ii) $i\in \mathfrak B$ and $j\in \mathfrak C$; (iii) $i\in \mathfrak C$ and $j\in \mathfrak D$; (iv) $i\in \mathfrak A$ and $j\in \mathfrak D$.

By Lemma \ref{mnmnmnmn} and in view of the definition of $A^{\Diamond}_i$ and $B^{\Diamond}_i$ 
$(\Diamond\in\{+, -\})$, we continue \eqref{hjyu09} as follows: 
\begin{align}
    & =\sum_{i_1,\dots,i_{m+n},i_{(m+n)'},\dots,i_{1'}\in\mathbb N}  c_{i_1 i_{1'}}^{-+}(\Delta_1)\cdots c_{i_m i_{m'}}^{-+}(\Delta_n)    \sum_{\nu\in\mathfrak R} (-1)^{\operatorname{Cross}(\nu)}  \quad c_{i_{m+1} i_{(m+1)'}}^{\Diamond_{m+1}(\nu) \Diamond_{(m+1)'}(\nu)}(\Delta_{m+1})
    \notag\\
    &\quad \times\dots\times  c_{i_{m+n} i_{(m+)'}}^{\Diamond_{m+n}(\nu) \Diamond_{(m+n)'}(\nu)}(\Delta_{m+n})
    \prod_{\substack{\{u,v\}\in\nu}}\delta_{i_u,i_v}.\label{tsw5wss}
\end{align}
Here, for $u\in \mathfrak E$, we denote $\Diamond_u(\nu):=-$ if  $\{u,v\}\in\nu$  and
$ u< v$, and $\Diamond_u(\nu):=+$ if $\{v,u\}\in\nu$ and $u<v$. \vspace{2mm}

{\it Step 3}. Let $I:\mathfrak E\to \mathfrak E$ be the bijective map that is defined as follows: $I$ acts as the identity on $\mathfrak A$ and $\mathfrak D$ and swaps the elements of $\mathfrak B$ and $\mathfrak C$, i.e.,  $I(i)=i'$ for all $i\in \mathfrak B$ and   $I(i')=i$ for all $i'\in \mathfrak C$.

Denote by $\mathfrak S$ the collection of all partitions $\nu$ of $\mathfrak E$ into $n$ two-point sets   such that, for each  $\{i,j\}\in\nu$, we have $i\in \mathfrak A\cup \mathfrak B$ and $j\in \mathfrak C\cup \mathfrak D$.
As easily seen, the map $I$ induces a bijection (still denoted by $I$) of $\mathfrak R$ onto $\mathfrak S$.

\begin{lemma}\label{d5t6dt}
(i) For $\nu\in\mathfrak R$, denote by $k_1(\nu)$ the number of all $\{i,j\}\in\nu$ such that $i\in \mathfrak A$, $j\in \mathfrak B$, and denote $k_2(\nu)$ the number of all $\{i,j\}\in\nu$ such that $i\in \mathfrak C$,  $j\in \mathfrak D$. Then 
  $k_1(\nu)=k_2(\nu)$ and we denote $  \mathfrak l(\nu):=k_1(\nu)=k_2(\nu)$.
    
  (ii) We have, for each $\nu\in\mathfrak R$ 
  \begin{equation}\label{vyre75i}(-1)^{\operatorname{Cross}(\nu)}=(-1)^{\operatorname{Cross}(I(\nu))}(-1)^{\mathfrak l(\nu)}.\end{equation} 
\end{lemma}

\begin{proof} Part (i) is obvious, so we only prove part (ii). 
For $\nu\in\mathfrak R$, under the map $I$, the change in the number of crossings happens in the following three cases:

\begin{itemize}
\item[(a)] Let $\{i_1,j_1\}, \{i_2,j_2\}\in \nu$, $i_1,i_2\in \mathfrak A$, $j_1, j_2\in \mathfrak B$.  Then $\{I(i_1),I(j_1)\}$, $\{I(i_2),I(j_2)\}$
  have a crossing if and only if $\{i_1, j_1\}$, $\{i_2,j_2\}$ do not have a crossing;

\item[(b)] Let $\{i_1,j_1\}, \{i_2,j_2\}\in \nu$, $i_1,i_2\in \mathfrak C$, $j_1, j_2\in \mathfrak D$. Then
  $\{I(i_1),I(j_1)\}$, $\{I(i_2),I(j_2)\}$
  have a crossing if and only if $\{i_1, j_1\}$, $\{i_2,j_2\}$ do not have a crossing;
  
  \item[(c)] Let $\{i_1,j_1\}, \{i_2,j_2\}\in \nu$, $i_1\in \mathfrak A$, $j_1\in \mathfrak B$, $i_2\in \mathfrak C$, $j_2\in \mathfrak D$. Then $\{i_1,j_1\}$, $\{i_2,j_2\}$ do not have a crossing while  
  $\{I(i_1),I(j_1)\}$ and $\{I(i_2),I(j_2)\}$  have a crossing.
\end{itemize}

Let $N_a(\nu)$, $N_b(\nu)$, and $N_c(\nu)$ denote the number of $\big\{\{i_1,j_1\}, \{i_2,j_2\}\big\}\subset \nu$ as in (a), (b), and (c), respectively. Since for any $i,j\in\mathbb N_0$, $(-1)^{i-j}=(-1)^{i+j}$, we therefore have:
\begin{equation}\label{vcyrds5}
(-1)^{\operatorname{Cross}(I(\nu))}=(-1)^{\operatorname{Cross}(\nu)
+N_{a}(\nu)+N_{b}(\nu)+N_c(\nu)}.\end{equation}
By part (i), $N_a(\nu)=N_b(\nu)$, while $N_c(\nu)=\mathfrak l(\nu)^2$. Hence, \eqref{vcyrds5} implies \eqref{vyre75i}. 
\end{proof}\vspace{2mm}
 
{\it Step 4}. We will now identify the set of partitions $\mathfrak R$  with the symmetric group  $  S_{m+n}$.
We define $\mathfrak m:\mathfrak C\cup\mathfrak D\to \mathfrak A\cup\mathfrak B$ by
$\mathfrak{ m}i'=i$ for $i'\in \mathfrak C\cup\mathfrak D$.

We define $\mathfrak I: \mathfrak R\to  S_{m+n}$
as follows: for $\nu\in \mathfrak R$ let $\xi=\mathfrak I(\nu)$ be given by:
\begin{itemize}
  \item for each $\{u,v\} \in \nu$ such that $u\in \mathfrak {A}$ and $v\in \mathfrak D$,
  $\xi(u)=\mathfrak{ m} v$,
  
  \item for each $\{u,v\}\in\nu $ such that $u\in \mathfrak {A}$ and  $v\in \mathfrak B$,
  $\xi(u)= v$,
  
  \item for each $\{u,v\}\in\nu $ such that $u\in \mathfrak {B}$ and  $v\in \mathfrak C$,
    $\xi(\mathfrak{ m}v)= u$,
    
     \item for each $\{u,v\}\in\nu $ such that $u\in \mathfrak {C}$ and $v\in \mathfrak D$,
       $\xi(\mathfrak{ m}u)=\mathfrak{m} v$.
\end{itemize}

As easily seen, for each $\nu\in\mathfrak R$, 
$$ \operatorname{Cross}(I(\nu))=\operatorname{sgn}(\mathfrak I(\nu)). $$
Then formula \eqref{vyre75i} implies
$$(-1)^{\operatorname{Cross}(\mathfrak I^{-1}(\xi))}=(-1)^{\operatorname{sgn}(\xi)}(-1)^{\mathfrak l(\mathfrak I^{-1}(\xi))},\quad \xi\in S_{m+n}.$$

Hence, we continue \eqref{tsw5wss} as follows:
\begin{align}
&=\sum_{\xi\in S_{m+n}}(-1)^{\operatorname{sgn}(\xi)}(-1)^{\mathfrak l(\mathfrak I^{-1}(\xi))} \sum_{i_1,\dots,i_{m+n},i_{(m+n)'},\dots,i_{1'}\in\mathbb N}
c_{i_1 i_{1'}}^{-+}(\Delta_1)\cdots c_{i_m i_{m'}}^{-+}(\Delta_n)\notag\\
&\quad\times c_{i_{m+1} i_{(m+1)'}}^{\Diamond_{m+1}(\mathfrak I^{-1}(\xi)) \Diamond_{(m+1)'}(\mathfrak I^{-1}(\xi))}(\Delta_{m+1})\dotsm c_{i_{m+n} i_{(m+n)'}}^{\Diamond_{m+n}(\mathfrak I^{-1}(\xi)) \Diamond_{(m+n)'}(\mathfrak I^{-1}(\xi))}(\Delta_{m+n})\notag\\
&\quad\times\prod_{\substack{\{u,v\}\in \mathfrak I^{-1}(\xi)}}\delta_{i_u,i_v}.\label{ftre7e4}
\end{align}\vspace{2mm}

{\it Step 5}. We define  mappings  $\mathfrak r_1:\mathfrak A\cup\mathfrak B\to \mathfrak A\cup\mathfrak C$ and $\mathfrak r_2:\mathfrak A\cup\mathfrak B\to \mathfrak B\cup\mathfrak D$ as follows:

\begin{itemize}

\item for $u\in\mathfrak A$, $\mathfrak r_1(u):=u$ and $\mathfrak r_2(u):=u'$;

\item for $u\in\mathfrak B$,  $\mathfrak r_1(u):=u'$ and $\mathfrak r_2(u):=u$.

\end{itemize}

Then, for each $\xi\in S_{m+n}$, we have $\{u,v\}\in \mathfrak I^{-1}(\xi)$ if and only if, for some $i\in\{1,\dots,m+n\}$, we have $\{u,v\}=\{\mathfrak r_1(u), \mathfrak r_2(\xi(u))\}$.
Hence, by \eqref{hjyu09}, \eqref{tsw5wss}, and \eqref{ftre7e4},
\begin{align}
& \tau\big({:}\rho(\Delta_1)\cdots \rho(\Delta_{m+n}){:}\big) \notag \\
&\quad =\sum_{\xi\in S_{m+n}}(-1)^{\operatorname{sgn}(\xi)}(-1)^{\mathfrak l(\mathfrak I^{-1}(\xi))} \sum_{i_1,\dots,i_{m+n},i_{(m+n)'},\dots,i_{1'}\in\mathbb N}
c_{i_1 i_{1'}}^{-+}(\Delta_1)\cdots c_{i_m i_{m'}}^{-+}(\Delta_m)\notag\\
&\qquad\times c_{i_{m+1} i_{(m+1)'}}^{\Diamond_{m+1}(\mathfrak I^{-1}(\xi)) \Diamond_{(m+1)'}(\mathfrak I^{-1}(\xi))}(\Delta_{m+1})\dotsm c_{i_{m+m} i_{(m+n)'}}^{\Diamond_{m+n}(\mathfrak I^{-1}(\xi)) \Diamond_{(m+n)'}(\mathfrak I^{-1}(\xi))}(\Delta_{m+n})\notag\\
&\qquad\times \prod_{i=1}^{m+n}\delta_{i_{\mathfrak r_1(u)},i_{\mathfrak r_2(\xi(u))}}.\label{utqwd6e}
\end{align}\vspace{2mm}

{\it Step 6}. Next, we prove 

\begin{lemma}\label{zdewq}
Let $\xi\in S_{m+n}$ and let $(l_1\,l_2\dotsm l_k)$ be a cycle in $\xi$. Then
\begin{align}\label{p98o0}
&\sum_{i_{l_1},i_{l_2},\dots,i_{l_k},i_{l_1'},i_{l_2'},\dots,i_{l_k'}
\in\mathbb N}\, c_{i_{l_1}i_{l_1'}}^{\Diamond_{l_1}(\mathfrak I^{-1}(\xi)) \Diamond_{l_1'}(\mathfrak I^{-1}(\xi))}
(\Delta_{l_1})
c_{i_{l_2}i_{l_2'}}^{\Diamond_{l_2}(\mathfrak I^{-1}(\xi)) \Diamond_{l_2'}(\mathfrak I^{-1}(\xi))}
(\Delta_{l_2})\notag\\
&\quad\times
\dotsm \times c_{i_{l_k}i_{l_k'}}^{\Diamond_{l_k}(\mathfrak I^{-1}(\xi)) \Diamond_{l_k'}(\mathfrak I^{-1}(\xi))}
(\Delta_{l_k})\delta_{i_{\mathfrak r_1(l_1)},i_{\mathfrak r_2(l_2)}}\delta_{i_{\mathfrak r_1(l_2)},i_{\mathfrak r_2(l_3)}}\dotsm \delta_{i_{\mathfrak r_1(l_k)},i_{\mathfrak r_2(l_1)}}\notag\\
&=\operatorname{Tr}\big(P_{\Delta_{l_k}} R(l_k,l_{k-1})P_{\Delta_{l_{k-1}}} R(l_{k-1},l_{k-2})P_{\Delta_{l_{k-2}}}\dotsm P_{\Delta_{l_1}}R(l_1,l_k)P_{\Delta_{l_k}}
\big).
\end{align}
Here, for $u,v\in\{1,2,\dots,m+n\}$,
\begin{equation}\label{yr75re7}
R(u,v):=\begin{cases}K,&\text{if }
\min\{u,v\}\le m,\\
\mathbf 1-K, &\text{if }
\min\{u,v\}\ge m+1.
\end{cases}\end{equation}
\end{lemma}

\begin{proof}
Let us first consider the case where $l_1,l_2,\dots,l_k\in\{1,\dots,m\}$.
Then the left-hand side of \eqref{p98o0} becomes \begin{align*}
    &\sum_{i_{l_1},i_{l_2},\dots,i_{l_k},i_{l_1'},i_{l_2'},\dots,i_{l_k'}\in\mathbb N}\, 
    c_{i_{l_1}i_{l_1'}}^{-+}(\Delta_{l_1})   c_{i_{l_2}i_{l_2'}}^{- +}(\Delta_{l_2})\dotsm   c_{i_{l_k}i_{l_k'}}^{-+}(\Delta_{l_k}) 
    \delta_{i_{ l_1},i_{l_2'}}\delta_{i_{ l_2},i_{l_3'}}\dots \delta_{i_{ l_k},i_{l_1'}}\\
    &\quad=\sum_{i_{l_1},i_{l_2},\dots,i_{l_k}\in\mathbb N}\,   c_{i_{l_1}i_{l_k}}^{-+}(\Delta_{l_1})   c_{i_{l_2}i_{l_1}}^{- +}(\Delta_{l_2})\dotsm   c_{i_{l_k}i_{l_{k-1}}}^{-+}(\Delta_{l_k}). 
\end{align*}
We have \begin{align*}
    &\sum_{ i_{l_1}\in\mathbb N}\, c_{i_{l_1}i_{l_k}}^{-+}(\Delta_{l_1})   c_{i_{l_2}i_{l_1}}^{- +}(\Delta_{l_2})  
      = \sum_{ i_{l_1}\in\mathbb N}\, \big(K_1P_{\Delta_{l_1}} K_1 e_{i_{l_k}},e_{i_{l_1}}\big)_\mathcal H    \big(K_1P_{\Delta_{l_2}} K_1 e_{i_{l_1}},e_{i_{l_2}}\big)_\mathcal H  \\
    &\quad = \big(K_1P_{\Delta_{l_2}}KP_{\Delta_{l_1}} K_1 e_{i_{l_k}},e_{i_{l_2}}\big)_\mathcal H  .
\end{align*}
Next,
\begin{align*}
    & \sum_{ i_{l_2}\in\mathbb N}\, \big(K_1P_{\Delta_{l_2}}KP_{\Delta_{l_1}} K_1 e_{i_{l_k}},e_{i_{l_2}}\big)_\mathcal H\, c_{i_{l_3}i_{l_2}}^{-+}(\Delta_{l_3}) \\
    &\quad  =\sum_{ i_{l_2}\in\mathbb N}\, \big(K_1P_{\Delta_{l_2}}KP_{\Delta_{l_1}} K_1 e_{i_{l_k}},e_{i_{l_2}}\big)_\mathcal H \big(K_1P_{\Delta_{l_3}} K_1 e_{i_{l_2}},e_{i_{l_3}}\big)_\mathcal H\\
    &\quad= \big(K_1P_{\Delta_{l_3}} K P_{\Delta_{l_2}}KP_{\Delta_{l_1}} K_1 e_{i_{l_k}},e_{i_{l_3}}\big)_\mathcal H .
\end{align*}
Continuing by analogy, we conclude : \begin{align*}
    &\sum_{i_{l_1},i_{l_2},\dots,i_{l_{k-1}}\in\mathbb N}\,  c_{i_{l_1}i_{l_k}}^{-+}(\Delta_{l_1})   c_{i_{l_2}i_{l_1}}^{-+}(\Delta_{l_2})\dots  c_{i_{l_k}i_{l_{k-1}}}^{-+}(\Delta_{l_k})    \\
    &\qquad
    =   \big(K_1P_{\Delta_{l_k}} K P_{\Delta_{l_{k-1}}}K\dotsm K P_{\Delta_{l_1}} K_1 e_{i_{l_k}},e_{i_{l_k}}\big)_\mathcal H ,
\end{align*}
which implies \begin{align*}
    & \sum_{i_{l_1},i_{l_2},\dots,i_{l_k}\in\mathbb N}\,  c_{i_{l_1}i_{l_k}}^{-+}(\Delta_{l_1})   c_{i_{l_2}i_{l_1}}^{-+}(\Delta_{l_2})\dotsm  c_{i_{l_k}i_{l_{k-1}}}^{-+}(\Delta_{l_k})  \\
    &\qquad= \sum_{ i_{l_k}\in\mathbb N}\,   \big(K_1P_{\Delta_{l_k}} K P_{\Delta_{l_{k-1}}}K\dotsm K P_{\Delta_{l_1}} K_1 e_{i_{l_k}},e_{i_{l_k}}\big)_\mathcal H \\
    &\qquad=\operatorname{Tr}\big(K_1P_{\Delta_{l_k}} K P_{\Delta_{l_{k-1}}}K\dotsm K P_{\Delta_{l_1}} K_1\big)\\
    &\qquad=\operatorname{Tr}\big(P_{\Delta_{l_k}} K P_{\Delta_{l_{k-1}}}K\dotsm K P_{\Delta_{l_1}} K\big)\\
    &\qquad=\operatorname{Tr}\big(P_{\Delta_{l_k}} K P_{\Delta_{l_{k-1}}}K\dotsm K P_{\Delta_{l_1}} K P_{\Delta_{l_k}}\big).
    \end{align*}
 
 We similarly treat the case where $l_1,l_2,\dots,l_k\in\{m+1,\dots,m+n\}$. 
 
 Finally we consider the case where 
$$
\{l_1,l_2,\dots,l_k\}\cap\{1,\dots,m\}\neq \varnothing,\quad 
\{l_1,l_2,\dots,l_k\}\cap\{m+1,\dots,m+n\}\neq \varnothing.
$$
Without loss of generality, we may assume that $l_1\in\{1,\dots,m\}$ and $l_k\in\{m+1,\dots,m+n\}.$ To simplify the notation, we will additionally assume that,
for some $\alpha\in\{1,\dots,k-1\}$, 
$$l_1,\dots,l_\alpha\in\{1,\dots,m\},\quad l_{\alpha+1},\dots,l_k\in\{m+1,\dots,m+n\}.$$
The interested reader can easily extend our arguments to the more general case.

We consider separately three cases:

{\it Case 1:} $\alpha=k-1$. Then the left-hand side of \eqref{p98o0} becomes 
 $$
    \sum_{i_{l_1},\dots,i_{l_{k-1}},i_{l_k'}\in\mathbb N}\,   c_{i_{l_1}i_{l_k'}}^{-+}(\Delta_{l_1})   c_{i_{l_2}i_{l_1}}^{- +}(\Delta_{l_2})\dotsm   c_{i_{l_{k-1}}i_{l_{k-2}}}^{-+}(\Delta_{l_{k-1}})c_{i_{l_{k-1}}i_{l_{k'}}}^{+-}(\Delta_{l_{k}}). 
$$
Analogously to the above calculations, we get
\begin{align*}
    &\sum_{i_{l_1},\dots,i_{l_{k-2}}\in\mathbb N}\, c_{i_{l_1}i_{l_k'}}^{-+}(\Delta_{l_1})   c_{i_{l_2}i_{l_1}}^{- +}(\Delta_{l_2})\dotsm c_{i_{l_{k-1}}i_{l_{k-2}}}^{-+}(\Delta_{l_{k-1}})   \\
    &\qquad=\big(K_1P_{\Delta_{l_{k-1}}}KP_{\Delta_{l_{k-2}}}K\dotsm KP_{\Delta_{l_1}}K_1e_{i_{l_k'}},e_{i_{l_{k-1}}}\big)_\mathcal H.
\end{align*}
Then
\begin{align*}
    &\sum_{i_{l_{k-1}}\in\mathbb N}\,\big(K_1P_{\Delta_{l_{k-1}}}KP_{\Delta_{l_{k-2}}}K\dotsm KP_{\Delta_{l_1}}K_1e_{i_{l_k'}},e_{i_{l_{k-1}}}\big)_\mathcal H\, c_{i_{l_{k-1}}i_{l_{k'}}}^{+-}(\Delta_{l_{k}})  \\
    & \quad= \sum_{i_{l_{k-1}}\in\mathbb N}\,\big(K_1P_{\Delta_{l_{k-1}}}KP_{\Delta_{l_{k-2}}}K\dotsm KP_{\Delta_{l_1}}K_1e_{i_{l_k'}},e_{i_{l_{k-1}}}\big)_\mathcal H
    \big(K_1P_{\Delta_{l_k}}K_1 e_{i_{l_{k-1}}},e_{i_{l'_{k}}}\big)_\mathcal H\\
    &\quad =\big(K_1P_{\Delta_{l_k}}KP_{\Delta_{l_{k-1}}}KP_{\Delta_{l_{k-2}}}K\dotsm KP_{\Delta_{l_1}}K_1e_{i_{l_k'}},e_{i_{l_k'}}\big)_\mathcal H,
    \end{align*}
and \begin{align*}
    &\sum_{i_{l_{k}'}\in\mathbb N}\,\big(K_1P_{\Delta_{l_k}}KP_{\Delta_{l_{k-1}}}KP_{\Delta_{l_{k-2}}}K\dotsm KP_{\Delta_{l_1}}K_1e_{i_{l_k'}},e_{i_{l_k'}}\big)_\mathcal H \\
&\quad =\operatorname{Tr}\big(P_{\Delta_{l_k}}KP_{\Delta_{l_{k-1}}}KP_{\Delta_{l_{k-2}}}K\dotsm KP_{\Delta_{l_1}}KP_{\Delta_{l_k}}\big).
\end{align*}

{\it Case 2:} $\alpha=k-2$. Then the left-hand side of \eqref{p98o0} becomes 
\begin{align*}
        &\sum_{i_{l_1},\dots,i_{l_{k-2}},i_{l_{k-1}'},i_{l_k'}\in\mathbb N}\,   c_{i_{l_1}i_{l_k'}}^{-+}(\Delta_{l_1})   c_{i_{l_2}i_{l_1'}}^{- +}(\Delta_{l_2})\dotsm  c_{i_{l_{k-2}}i_{l_{k-3}}}^{-+}(\Delta_{l_{k-2}})\\
    &\quad\qquad \times
      c_{i_{l_{k-2}}i_{l_{k-1}'}}^{++}(\Delta_{l_{k-1}})c_{i_{l_{k-1}'}i_{l_{k}'}}^{--}(\Delta_{l_{k}})\\
    &\qquad=\sum_{i_{l_{k-2}},i_{l_{k-1}'},i_{l_k'}\in\mathbb N}\,\big(K_1P_{\Delta_{l_{k-2}}}KP_{\Delta_{l_{k-3}}}K\dotsm KP_{\Delta_{l_1}}K_1e_{i_{l_k'}},e_{i_{l_{k-2}}}\big)_\mathcal H\\
    &\quad\qquad\times \big(K_2P_{\Delta_{l_{k-1}}}K_1 e_{i_{l_{k-2}}},e_{i_{l_{k-1}'}}\big)_\mathcal H  \big(K_1P_{\Delta_{l_{k}}}K_2 e_{i_{l_{k-1}'}},e_{i_{l_k'}}\big)_\mathcal H\\
  &\qquad=\sum_{i_{l_{k-1}'},i_{l_k'}\in\mathbb N}\, \big(K_2P_{\Delta_{l_{k-1}}}P_{\Delta_{l_{k-2}}}KP_{\Delta_{l_{k-3}}}K\dotsm KP_{\Delta_{l_1}}K_1e_{i_{l_k'}},e_{i_{l_{k-1}'}}\big)_\mathcal H
  \\
  &\quad\qquad \times 
  \big(K_1P_{\Delta_{l_{k}}}K_2 e_{i_{l_{k-1}'}},e_{i_{l_k'}}\big)_\mathcal H\\
 &\qquad=\sum_{i_{l_k'}\in\mathbb N}\,\big( K_1 P_{\Delta_{l_{k}}} (\mathbf{ 1}-K) P_{\Delta_{l_{k-1}}} K P_{\Delta_{l_{k-2}}}\dotsm KP_{\Delta_{l_1}}K_1e_{i_{l_k'}},e_{i_{l_k'}}\big)_\mathcal H\\
 &\qquad =\operatorname{Tr}\big(P_{\Delta_{l_{k}}} (\mathbf{ 1}-K) P_{\Delta_{l_{k-1}}} K P_{\Delta_{l_{k-2}}}K\dotsm KP_{\Delta_{l_1}}KP_{\Delta_{l_{k}}} \big).
\end{align*}

{ \it  Case 3:} $\alpha\leqslant k-3$. Then the left-hand side of \eqref{p98o0} becomes 
\begin{align*}
& =\sum_{i_{l_1},\dots,i_{l_\alpha},i_{l_{\alpha+1}'},i_{l_{\alpha+2}'},\dots,i_{l_k'}\in\mathbb N}\,c_{i_{l_1}i_{l_k'}}^{-+}(\Delta_{l_1})c_{i_{l_2}i_{l_1}}^{-+}(\Delta_{l_2})  \dotsm  c_{i_{l_\alpha}i_{l_{\alpha-1}}}^{-+}(\Delta_{l_\alpha}) 
 \\
 &\quad  \times c_{i_{l_\alpha}i_{l_{\alpha+1}'}}^{++}(\Delta_{l_{\alpha+1}})
 c_{i_{l_{\alpha+1}'}i_{l_{\alpha+2}'}}^{-+}(\Delta_{l_{\alpha+2}})\dotsm c_{i_{l_{k-2}'}i_{l_{k-1}'}}^{-+}(\Delta_{l_{k-1}})c_{i_{l_{k-1}'}i_{l_{k}'}}^{--} (\Delta_{l_{k}}).   
\end{align*}
Similarly to Case 2, we obtain:
\begin{align*}
    &\sum_{i_{l_1},\dots,i_{l_\alpha}\in\mathbb N}\, c_{i_{l_1}i_{l_k'}}^{-+}(\Delta_{l_1})c_{i_{l_2}i_{l_1}}^{-+}(\Delta_{l_2})  \dotsm  c_{i_{l_\alpha}i_{l_{\alpha-1}}}^{-+}(\Delta_{l_\alpha}) 
 c_{i_{l_\alpha}i_{l_{\alpha+1}'}}^{++}(\Delta_{l_{\alpha+1}})  \\
    &\quad =\big(K_2P_{\Delta_{l_{\alpha+1}}}  KP_{\Delta_{l_{\alpha}}}K\dotsm KP_{\Delta_{l_{1}}}K_1 e_{i_{l_k'}},e_{i_{l_{\alpha+1}'}}\big)_\mathcal H.
\end{align*}
Then \begin{align*}
    & \sum_{i_{l_{\alpha+1}'},\dots,i_{l_{k-2}'}\in\mathbb N}\, \big(K_2P_{\Delta_{l_{\alpha+1}}}  KP_{\Delta_{l_{\alpha}}}K\dotsm KP_{\Delta_{l_{1}}}K_1 e_{i_{l_k'}},e_{i_{l_{\alpha+1}'}}\big)_\mathcal H  \\
    &\qquad \times c_{i_{l_{\alpha+1}'}i_{l_{\alpha+2}'}}^{-+}(\Delta_{l_{\alpha+2}})\dotsm c_{i_{l_{k-2}'}i_{l_{k-1}'}}^{-+}(\Delta_{l_{k-1}})\\
    &\quad =  \sum_{i_{l_{\alpha+1}'},\dots,i_{l_{k-2}'}\in\mathbb N}\, \big(K_2P_{\Delta_{l_{\alpha+1}}}  KP_{\Delta_{l_{\alpha}}}K\dotsm KP_{\Delta_{l_{1}}}K_1 e_{i_{l_k'}},e_{i_{l_{\alpha+1}'}}\big)_\mathcal H  \\
    &\qquad \times \big(K_2 P_{\Delta_{l_{\alpha+2}}}K_2 e_{i_{l_{\alpha+1}'}},e_{i_{l_{\alpha+2}'}}\big)_\mathcal H\dotsm \big(K_2 P_{\Delta_{l_{k-1}}}K_2 e_{i_{l_{k-2}'}},e_{i_{l_{k-1}'}}\big)_\mathcal H\\
    &\quad = \big(K_2 P_{\Delta_{l_{k-1}}}(\mathbf{ 1}-K) P_{\Delta_{l_{k-2}}}(\mathbf{ 1}-K) \dotsm (\mathbf{ 1}-K) P_{\Delta_{l_{\alpha+1}}}KP_{\Delta_{l_{\alpha}}}K\\
    &\qquad \dotsm KP_{\Delta_{l_{1}}}K_1 e_{i_{l_{k}'}},e_{i_{l_{k-1}'}}\big)_\mathcal H.
\end{align*}
Finally, 
\begin{align*}
    &  \sum_{i_{l_{k-1}'},i_{l_k'}\in\mathbb N}\, \big(K_2 P_{\Delta_{l_{k-1}}}(\mathbf{ 1}-K) P_{\Delta_{l_{k-2}}}(\mathbf{ 1}-K) \dotsm (\mathbf{ 1}-K) P_{\Delta_{l_{\alpha+1}}}KP_{\Delta_{l_{\alpha}}}K\\
    &\quad  \dotsm KP_{\Delta_{l_{1}}}K_1 e_{i_{l_{k}'}},e_{i_{l_{k-1}'}}\big)_\mathcal H\, c_{i_{l_{k-1}'}i_{l_{k}'}}^{--} (\Delta_{l_{k}}) \\
    & = \sum_{i_{l_{k-1}'},i_{l_k'}\in\mathbb N}\, \big(K_2 P_{\Delta_{l_{k-1}}}(\mathbf{ 1}-K) P_{\Delta_{l_{k-2}}}(\mathbf{ 1}-K) \dotsm (\mathbf{ 1}-K) P_{\Delta_{l_{\alpha+1}}}KP_{\Delta_{l_{\alpha}}}K\\
    & \quad \dotsm KP_{\Delta_{l_{1}}}K_1 e_{i_{l_{k}'}},e_{i_{l_{k-1}'}}\big)_\mathcal H \big(K_1 P_{\Delta_{l_{k}}}K_2 e_{i_{l_{k-1}'}},e_{i_{l_{k}'}}\big)_\mathcal H\\
  & =\operatorname{Tr}\big( P_{\Delta_{l_{k}}}(\mathbf{ 1}-K)P_{\Delta_{l_{k-1}}}(\mathbf{ 1}-K)\dotsm (\mathbf{ 1}-K) P_{\Delta_{l_{\alpha+1}}}K P_{\Delta_{l_\alpha}} K 
  \dotsm K P_{\Delta_{l_1}}K P_{\Delta_{l_{k}}}\big)_\mathcal H.\quad 
\end{align*}\vspace{2mm}

{\it Step 7}. For a given cycle $\theta= (l_1 l_2\dotsm l_k)$ in a permutation $\xi\in S_{m+n}$, we denote by $\widetilde {\mathbb T}_\theta$ the value given by (the right hand-side of) formula \eqref{p98o0}. Denote by $\mathfrak{t}(\theta)$ the number of $i\in\{1,\dots,k\}$ such that $l_i\in\{m+1,\dots,m+n\}$ but $l_{i+1}\in\{1,\dots,m\}$, where  $l_{k+1}:=l_1$. Then, by \eqref{vtyde6e6x}, \eqref{dsts5a}, and \eqref{yr75re7}, 
\begin{equation}
\label{monamona}
(-1)^{\mathfrak{t}(\theta)}\,\widetilde{\mathbb T}_\theta=\operatorname{Tr}\big(P_{\Delta_{l_{k}}} \mathbb K P_{\Delta_{l_{k-1}}} \mathbb K \dotsm \mathbb K P_{\Delta_{l_1}}\mathbb K P_{\Delta_{l_{k}}}\big)
=\mathbb{T}_{\theta^{-1}}\,.\end{equation}
By Lemma~\ref{d5t6dt} (i),
  \begin{equation}\label{srq5y3u}
  \sum_{\theta\in \operatorname{Cycles}(\xi)}\mathfrak{t}(\theta)= \mathfrak{l}(\mathfrak{I}^{-1}(\xi)).\end{equation}

Thus, by \eqref{utqwd6e}, \eqref{monamona}, \eqref{srq5y3u} and Lemma~\ref{zdewq},
\begin{align}
     \tau\big({:}\rho(\Delta_1)\cdots \rho(\Delta_{m+n}){:}\big) 
&=\sum_{\xi\in S_{m+n}}(-1)^{\operatorname{sgn}(\xi)}\prod_{\theta\in \operatorname{Cycles}(\xi)} \mathbb{T}_{\theta^{-1}}\notag\\
&=\sum_{\xi\in S_{m+n}}(-1)^{\operatorname{sgn}(\xi)}\prod_{\theta\in \operatorname{Cycles}(\xi)} \mathbb{T}_{\theta}.\label{ctewu5643u}
 \end{align}
 Formulas \eqref{vcyrd6de} and \eqref{ctewu5643u} imply\eqref{vcyrtse6u785}.
  \end{proof}
  
   \section*{Acknowledgements} 
   EL is grateful to  Grigori Olshanski for many useful discussions on determinantal point processes with $J$-Hermitian correlation kernel.  The authors are grateful to  the anonymous reviewers for many useful comments and suggestions.

\end{document}